\documentclass[final]{ectj}
\usepackage{amsmath,amssymb,booktabs,array,longtable,tabularx,microtype,xurl}
\usepackage{placeins}
\usepackage[colorlinks=true,linkcolor=black,citecolor=blue,urlcolor=blue]{hyperref}
\usepackage{pgfplots}
\pgfplotsset{compat=1.18}
\newcommand{\IS}{\mathrm{IS}_{90}}
\newcommand{\notes}[1]{\par\smallskip{\footnotesize\setlength{\baselineskip}{11pt}\raggedright #1\par}}
\newtheorem{proposition}{Proposition}
\newtheorem{example}{Example}
\newtheorem{corollary}{Corollary}
\makeatletter
\@addtoreset{proposition}{section}
\@addtoreset{example}{section}
\@addtoreset{corollary}{section}
\makeatother

\title[Revised data and predictive uncertainty]{Predictive Uncertainty with Revised Economic Data}
\author[Kuang]{Yizhou (Kyle) Kuang$^{\dagger}$}
\address{$^{\dagger}$University of Manchester.}
\email{yizhou.kuang@manchester.ac.uk}
\usepackage{etoolbox}
\hypersetup{pdftitle={Predictive Uncertainty with Revised Economic Data},pdfauthor={Yizhou (Kyle) Kuang}}
\makeatletter
\let\ArxivMakeTitle\maketitle
\let\ArxivTitle\@title
\let\ArxivAuthors\@authors
\let\ArxivThanks\thanks
\let\ArxivNormalSectionFormat\@seccntformat
\def\ArxivFirstPageLabel{arxiv-main-firstpage}
\def\ArxivLastPageLabel{arxiv-main-lastpage}
\def\@submissionheader{Preprint}
\patchcmd{\ps@title}{\pageref{firstpage}}{\pageref{\ArxivFirstPageLabel}}{}{\PackageError{arxiv-package}{First-page header patch failed}{}}
\patchcmd{\ps@title}{\pageref{lastpage}}{\pageref{\ArxivLastPageLabel}}{}{\PackageError{arxiv-package}{Last-page header patch failed}{}}
\patchcmd{\HyOrg@maketitle}{\textsc{online appendix}}{\textsc{Online Appendix}}{}{\PackageError{arxiv-package}{Online Appendix title patch failed}{}}
\def\label@firstpage{\phantomsection\label{\ArxivFirstPageLabel}\ifectj@onlineappendix\else\pdfbookmark[0]{Predictive Uncertainty with Revised Economic Data}{arxiv-paper}\fi}
\def\label@lastpage{}
\newcommand{\StartOnlineAppendix}{%
 \par\phantomsection\label{arxiv-main-lastpage}\clearpage
 \setcounter{page}{1}%
 \setcounter{section}{0}%
 \setcounter{subsection}{0}%
 \setcounter{subsubsection}{0}%
 \setcounter{equation}{0}%
 \setcounter{figure}{0}%
 \setcounter{table}{0}%
 \setcounter{proposition}{0}%
 \setcounter{example}{0}%
 \setcounter{corollary}{0}%
 \setcounter{footnote}{0}%
 \renewcommand{\thesection}{S\arabic{section}}%
 \renewcommand{\thepage}{S\arabic{page}}%
 \renewcommand{\thefigure}{S\arabic{figure}}%
 \renewcommand{\thetable}{S\arabic{table}}%
 \renewcommand{\theequation}{\thesection.\arabic{equation}}%
 \def\theHsection{oa.\arabic{section}}%
 \def\theHsubsection{\theHsection.\arabic{subsection}}%
 \def\theHsubsubsection{\theHsubsection.\arabic{subsubsection}}%
 \def\theHequation{\theHsection.\arabic{equation}}%
 \def\theHfigure{oa.\arabic{figure}}%
 \def\theHtable{oa.\arabic{table}}%
 \def\theHproposition{\theHsection.\arabic{proposition}}%
 \def\theHexample{\theHsection.\arabic{example}}%
 \def\theHcorollary{\theHsection.\arabic{corollary}}%
 \def\theHfootnote{oa.\arabic{Hfootnote}}%
 \@ifundefined{c@proof}{}{\setcounter{proof}{0}\def\theHproof{oa.\arabic{proof}}}%
 \let\@seccntformat\ArxivNormalSectionFormat
 \def\ArxivFirstPageLabel{arxiv-oa-firstpage}%
 \def\ArxivLastPageLabel{arxiv-oa-lastpage}%
 \def\@submissionheader{Online Appendix}%
 \ectj@onlineappendixtrue
 \let\@title\ArxivTitle
 \let\@authors\ArxivAuthors
 \let\thanks\ArxivThanks
 \let\maketitle\ArxivMakeTitle
 \interfootnotelinepenalty=100
 \pdfbookmark[0]{Online Appendix}{arxiv-online-appendix}
 \maketitle
}
\makeatother
\begin{document}

\begin{abstract}
Data revisions create uncertainty that outlives the first publication of an economic indicator. We develop a framework that turns an existing point forecast of the first release into a predictive distribution for a later release. The framework combines forecast-error records with official-release histories, including periods without recorded forecasts. We establish identification conditions and, from the component distributions, construct sharp bounds for directional probabilities and dependence-robust prediction intervals. We also show when additional revision records strictly improve estimation precision for a fixed rare-tail bound. Monte Carlo simulations quantify the predictive gains from longer release histories. Across four US inflation series, full histories improve interval scores relative to forecast-matched histories, with the largest gains in PCE. The conditional procedures lower pooled interval scores by 27--28 percent relative to a joint stochastic-volatility benchmark. With recency weighting, they match its observed coverage using intervals 34--36 percent narrower.
\end{abstract}
\section{Introduction}\label{main-sec:introduction}

Data revisions can overturn an economic assessment. An apparent slowdown in inflation can become an acceleration in a later publication. Policymakers and firms must nevertheless act before those revisions are known, judging both how the data may change and how strongly the available information supports their assessment.

Historical first and later releases reveal how initial readings change, even for periods without a recorded forecast. More recent forecasts paired with first releases reveal initial errors before the corresponding revisions arrive. Drawing on these complementary records therefore expands the information available beyond completed forecast-error/revision pairs. In our PCE application, the median histories contain 53 initial errors and 258 revisions, compared with only 49 completed pairs.

This paper provides a framework to identify and estimate the later-release predictive distribution from these complementary histories. We estimate revisions conditional on the first release, then integrate over the possible first releases implied by the supplied forecast and its error distribution. We establish conditions for identification without prescribing a parametric form for the conditional revision distribution. Because the construction accepts point forecasts from any forecasting system, institutions can apply it to their existing forecasts. Comparing full and forecast-matched histories within the same conditional specification isolates what the additional release records contribute. Monte Carlo simulations quantify this gain and the role of a transferable conditional revision relationship.

Our framework also answers a question central to economic decisions: how strongly does the available information support a conclusion such as inflation deceleration? Taking the initial-error and revision distributions as inputs, we use sharp bounds to determine which probability requirements hold across the admissible dependence patterns. We characterize how component-distribution errors affect those conclusions and show when additional revision observations strictly improve estimation precision for a fixed rare-tail bound. Dependence-robust prediction intervals carry this logic into prediction, making stronger tail requirements explicit in wider intervals and fewer directional declarations.

The benchmark comparisons show that this forecast augmentation can be valuable even against complete revision-aware forecasting systems. Across four inflation series, the linear and spline conditional procedures lower pooled interval scores by 27--28 percent relative to the joint stochastic-volatility benchmark (joint SV) based on \citet{Clements2023}. With recency weighting, both cover 124 of the 134 outcomes, as does joint SV, using intervals 34--36 percent narrower. In core PCE, the linear procedure covers 18 of 21 outcomes with no severe misses, compared with 12 and three severe misses for the real-time-vintage benchmark based on \citet{Clements2017}. The largest information gains come from avoiding costly PCE misses through better interval placement.

\citet{Clements2017} studies predictive uncertainty when evaluation data are revised, aligning estimation vintages with forecast-date inputs. \citet{Clements2023} model successive releases jointly in a Bayesian system with time-varying volatility. Both provide close revision-aware alternatives for predicting a later publication. Our starting point is an existing forecast that an institution wishes to retain. We use a separate conditional revision relationship to augment that forecast with official release records, including periods absent from its forecast archive. The current forecast-error distribution and the revision-training history therefore enter separately, allowing each to use the records it requires. This separation also permits a direct full-versus-matched comparison that holds the forecasting inputs and revision specification fixed. The accompanying confidence analysis asks a further question: which economic assertions follow from the component distributions when their dependence is only partially specified? It connects the information used for prediction to the accuracy, precision and dependence robustness required for a directional conclusion.

The broader real-time-data literature shows why forecast-date information and the release chosen for evaluation matter \citep{Croushore2001,Croushore2011,Stark2010,Clements2019}. Revision regressions and maturity-adapted autoregressions exploit release information \citep{Croushore2019,Clements2013}. Measurement systems distinguish news from noise in successive releases \citep{Jacobs2011,Kishor2012}, while density-nowcasting research combines real-time inflation forecasts \citep{Knotek2023}.

\section{Prediction problem and information}\label{main-sec:setup}

Let $t$ denote a reference period, $A_t$ the first official estimate of an economic quantity, and $B_t$ its estimate in a specified later publication. Both releases concern the same quantity and reference period. The quantity can be a growth rate, a level or a change in an economic indicator. At forecast date $s$, an analyst has information $\mathcal I_s$ and a supplied point forecast $m_s$ of $A_t$. Define the initial forecast error $E_t$ and cumulative revision $D_t$ by
\begin{equation}\label{main-eq:identity}
 E_t=A_t-m_s,\qquad D_t=B_t-A_t,\qquad B_t=m_s+E_t+D_t.
\end{equation}
The revision $D_t$ includes every intervening update between the first release and the chosen later publication. A threshold statement requires $P(B_t<b\mid\mathcal I_s)\ge1-\gamma$, where $b$ is a benchmark known at the forecast date and $1-\gamma$ is the required probability. For example, a zero benchmark indicates contraction when $B_t$ is output growth, or slowing when it is a change in inflation. After the first release, the analyst can also forecast $B_t$ conditional on the observed $A_t$.

An observation enters the information set once its required publications are available. Let $j$ index historical reference periods. For each period with a recorded forecast, let $s_j$ denote its forecast date and $m_{s_j}$ the forecast of $A_j$. Its initial error is $E_j=A_j-m_{s_j}$. Let $\mathcal H_s^E$ contain the initial errors known at $s$, $\mathcal H_s^D$ the known revisions, and $\mathcal H_s^P$ the completed $(E_j,D_j)$ pairs. These histories need not cover the same periods. A forecast and its first release reveal $E_j$ before the later revision is known. Official first and later releases reveal $(A_j,D_j)$ even when no historical forecast was recorded. A completed $(E_j,D_j)$ pair requires both releases and the historical forecast. The initial-error and revision histories can therefore include observations unavailable to an analysis restricted to completed pairs.

\section{Conditional prediction from complementary histories}\label{main-sec:prediction}

Historical release records reveal how initial readings are revised, including periods without recorded forecasts. We combine this relationship with the supplied forecast and its error distribution to predict the later release.

\subsection{Joint release records}

To allow the revision relationship to vary seasonally, let $c\in\{1,\ldots,12\}$ denote a month of the year. Write $c_j$ for the month of historical reference period $j$ and $c_t$ for the month of the current reference period $t$.\footnote{A May forecast for June has $c_t=6$. Historical June observations also have $c_j=6$.} Let $G(d\mid a,c)$ denote the conditional probability that the revision is at most $d$, given a first-release value $a$ for a reference period in month $c$.

Among historical records eligible at forecast date $s$, we assume that the conditional revision distribution is the same across dates, data sources and record-availability patterns once the first-release value and calendar month are held fixed. We further assume that this common conditional law applies to the current observation,
\begin{equation}\label{main-eq:alignment}
 P(D_t\le d\mid A_t=a,\mathcal I_s)=G(d\mid a,c_t).
\end{equation}
These restrictions hold for every $d$ and almost every $a$ under the respective first-release distributions.

\begin{samepage}
\begin{proposition}[Identification from joint release records]\label{main-prop:bridge}
Suppose the historical common-law restriction above and Equation~(\ref{main-eq:alignment}) hold at forecast date $s$, and the current forecast-error CDF $F_{E,s}(e)=P(E_t\le e\mid\mathcal I_s)$ is separately identified. The implied first-release CDF is $F_{A,s}(a)=F_{E,s}(a-m_s)$. Assume that the current first-release law is absolutely continuous with respect to a mixture assigning positive weight to each first-release law from eligible historical joint $(A,D)$ records with $c_j=c_t$. Then the current later-release CDF $F_{B,s}(b)=P(B_t\le b\mid\mathcal I_s)$ is identified by
\begin{equation}\label{main-eq:population-bridge}
 F_{B,s}(b)=\int G(b-m_s-e\mid m_s+e,c_t)\,dF_{E,s}(e).
\end{equation}
\end{proposition}
\end{samepage}

Absolute continuity identifies the historical conditional revision law for almost every first-release value under the current law. The observed sample range provides a diagnostic for extrapolation. Predictions where domination fails require an additional extrapolation assumption. The result allows nonparametric forms for both $F_{E,s}$ and $G$.

For implementation, we introduce a linear conditional revision model, $D=\eta_c+\beta_A A+u$, with month-specific intercept $\eta_c$, first-release slope $\beta_A$ and regression error $u$. To measure the contribution of additional release records, we fit the model twice. The \emph{full history} uses all eligible joint release records, while the \emph{matched history} uses the subset with recorded historical forecasts. Training observations receive equal weight. We use a Student predictive distribution to account for uncertainty in the estimated coefficients and residual variance under the Gaussian regression model.

To check whether the gains from additional release records depend on a linear revision mean, we also consider a spline conditional procedure. It extends the linear model parsimoniously with a natural cubic spline in the first release.\footnote{Four fixed knots add two nonlinear basis terms to the linear model, retaining the month intercepts. Full and matched histories share knots selected from their common presample, using releases published before 2022 in the application.}\label{main-sec:richer-models}

Let $\widehat G_s$ denote the fitted conditional revision CDF, and let $\widehat F_{E,s}$ denote the initial-error CDF supplied or estimated at forecast date $s$. Before the first release, $A_t$ is unknown.\footnote{After observing $A_t=a$, the same conditional law gives $P(B_t\le b\mid\mathcal I_s,A_t=a)=G(b-a\mid a,c_t)$.} For each candidate error $e$, form the hypothetical first release $m_s+e$, predict its revision, and add the two quantities. The resulting estimated later-release CDF is
\begin{equation}\label{main-eq:bridge}
 \widehat F_{B,s}(b)=\int\widehat G_s(b-m_s-e\mid m_s+e,c_t)
                         \,d\widehat F_{E,s}(e).
\end{equation}
The full and matched fits share the forecast, initial-error CDF, specification and evaluation dates.

\subsection{Predictive distributions}

Conditional revision learning translates the supplied forecast into a distribution for the later release, allowing its location to reflect both initial forecast errors and subsequent revisions. In the linear specification, the mean initial error and the conditional revision slope both affect the later-release predictive mean, which need not coincide with the supplied forecast. Implementation details for the linear and spline procedures, including their predictive means, are in Section~\ref{oa-sec:oa-conditional} of the Online Appendix.

A stable conditional revision law can generate a changing marginal revision distribution as the first-release distribution changes. Let $F_{D,s}(d)=P(D_t\le d\mid\mathcal I_s)$. For the linear model above, with finite first moments and $\mathbb E(u\mid\mathcal I_s)=0$,
\begin{equation}\label{main-eq:induced-margin}
 \begin{split}
 F_{D,s}(d)&=\int G(d\mid a,c_t)\,dF_{A,s}(a),\\
 \mathbb E(D_t\mid\mathcal I_s)&=\eta_{c_t}+\beta_A\mathbb E(A_t\mid\mathcal I_s).
 \end{split}
\end{equation}
The component approach developed in Section~\ref{main-sec:confidence} uses a separately supplied revision marginal. Substituting $A=m_s+E$ in the linear model gives $D=\eta_{c_t}+\beta_A m_s+\beta_A E+u$. A pooled regression of $D$ on $E$ with a constant intercept therefore generally imposes a different transfer restriction when the forecast varies.

\section{Dependence-sensitive confidence and robust prediction}\label{main-sec:confidence}

Accurate initial-error and revision distributions can still leave uncertainty about how the two errors combine. Revisions may reinforce an initial economic surprise, as when BEA's 2011 update deepened the estimated 2007--2009 GDP contraction from 4.1 to 5.1 percent \citep{BureauOfEconomicAnalysis2011}, or reverse an apparent expansion into a contraction, as in the first-quarter 2014 US GDP revision \citep{BureauOfEconomicAnalysis2014}. We express directional conclusions using $B_t<0$.\footnote{Any threshold known at the forecast date can be reduced to zero by subtracting it from the forecast and both releases. This leaves $E$ and $D$ unchanged.}

\begin{example}[Revision amplification]\label{main-ex:amplification}
\normalfont\par\noindent
At a fixed forecast $m$, write $A=m+E$ and $B=m+E+D$. Suppose $(E,D)$ is jointly Gaussian and $0\le\kappa\le1$ indexes their association. The component laws and predictive variances satisfy
\begin{equation}\label{main-eq:amplification}
 \begin{split}
 E,D&\sim N(0,1),\qquad \rho_\kappa=\operatorname{Cov}(E,D)=0.9\kappa-0.1,\\
 D\mid E=e&\sim N(\rho_\kappa e,1-\rho_\kappa^2),\qquad
 \operatorname{Var}(B-m)=1.8(1+\kappa).
 \end{split}
\end{equation}
Increasing $\kappa$ from zero to one doubles later-release variance while conditional revision variance falls from $0.99$ to $0.36$. The first release becomes more informative about its revision, yet uncertainty before the first release increases. Both component distributions remain unchanged.
\end{example}

\subsection{Directional confidence}\label{main-sec:dependence-bounds}

Complementary histories supply observations for each component distribution, while completed pairs may give limited guidance about current dependence. Holding the component distributions fixed allows us to distinguish directional conclusions supported under every dependence pattern from those requiring a particular assumption.

Let $\widehat\mu$ and $\widehat\nu$ denote the supplied marginal distributions of $E$ and $D$. The revision marginal is supplied separately from the conditional model in Section~\ref{main-sec:prediction}. A \emph{coupling} is a joint distribution with these margins, and $\Pi_s=\Pi(\widehat\mu,\widehat\nu)$ is the collection of all such distributions at forecast date $s$. Under independence, the joint distribution is the product of the two marginal distributions. Other couplings preserve the component probabilities while changing how initial errors and revisions occur together.

The resulting bounds build on the fixed-marginal sum-distribution bounds of \citet{Makarov1982} and \citet{Ruschendorf1982}. \citet{Zhang2025} revisit their sharpness and the treatment of discontinuous distributions.

Write $g(e,d)=m_s+e+d$. For a reference law $Q^{\rm ref}\in\Pi_s$, let $p^{\rm ref}=Q^{\rm ref}\{g(E,D)\ge0\}$ denote the probability that inflation does not slow under a specified dependence assumption. Let $[\underline p,\overline p]$ be the sharp range of this event probability across all couplings of the supplied component distributions. At $\gamma=0.05$ or $0.10$, deceleration meets the $1-\gamma$ requirement under the reference law when $p^{\rm ref}\le\gamma$, and under every dependence pattern when $\overline p\le\gamma$. Equality qualifies. This one-sided requirement differs from whether an equal-tail interval contains zero.

For finite $E$ support points $e_k$ in increasing order, the largest and smallest non-deceleration probabilities are
\begin{equation}\label{main-eq:sharp-probabilities}
\begin{aligned}
 \overline p&=\min_k[\widehat\mu(E>e_k)+\widehat\nu\{g(e_k,D)\ge0\}],\\
 \underline p&=1-\min_k[\widehat\mu(E<e_k)+\widehat\nu\{g(e_k,D)<0\}].
\end{aligned}
\end{equation}
The bounds are sharp. Joint distributions attain both extremes and every probability between them. Thus $p^{\rm ref}\le\gamma<\overline p$ identifies a dependence-sensitive conclusion supported by the reference law.

For the continuous Gaussian laws in Example~\ref{main-ex:amplification}, at $m=-2.4$, independence assigns about 96 percent probability to slowing, versus just under 90 percent under $\kappa=1$. The dependence bound certifies 77 percent and withholds the 95 percent assertion. At $m=-4$, even that bound exceeds 95 percent. A clear enough signal therefore supports the conclusion without specifying dependence. \hyperref[oa-app:proof-gaussian]{Section~\ref*{oa-sec:gaussian-protection}} of the Online Appendix gives the shock construction, Gaussian calculations and interval comparisons.

\subsection{Dependence-robust prediction intervals}\label{main-sec:protection}

We apply the same component information to both tails. With the supplied distributions held fixed, a dependence-robust prediction interval $R=[l_s,u_s]$ satisfies
\begin{equation}\label{main-eq:robust-tails}
 \sup_{Q\in\Pi_s}Q(B<l_s)\le 0.05,\qquad
 \sup_{Q\in\Pi_s}Q(B>u_s)\le 0.05.
\end{equation}
Every law in $\Pi_s$ then assigns at least 0.90 probability to the interval. For finite supports $\{e_i\}$ of $\widehat\mu$ and $\{d_j\}$ of $\widehat\nu$, define $\mathcal S_s=\{m_s+e_i+d_j\}$. The ideal rule uses the largest admissible lower and smallest admissible upper endpoint,
\[
\begin{aligned}
 l_R&=\max\{l\in\mathcal S_s:\sup_{Q\in\Pi_s}Q(B<l)\le 0.05\},\\
 u_R&=\min\{u\in\mathcal S_s:\inf_{Q\in\Pi_s}Q(B\le u)\ge 0.95\}.
\end{aligned}
\]
Within this finite-support construction, the ideal endpoints give the narrowest interval satisfying the two separate tail budgets. The upper cumulative condition is equivalent to the strict outside-tail budget $B>u$.

The class of admitted laws expresses the information the analyst is willing to maintain. Credible restrictions on dependence can narrow the class and reduce interval width. The proposition below justifies the endpoint construction through expected-score optimality among intervals subject to the same two tail requirements.

For nominal coverage $1-\alpha$, with $0<\alpha<1$, set $\tau=\alpha/2$. The interval score for $[l,u]$ is
\begin{equation}\label{main-eq:score}
 IS_\alpha(l,u;B)=(u-l)+\tau^{-1}\{(l-B)_++(B-u)_+\}.
\end{equation}
Here $(z)_+=\max(z,0)$. The score adds interval width to a penalty for distance outside the interval \citep{Gneiting2007}.

\begin{samepage}
Consider a nonempty, uniformly tight class $\mathcal Q$, each of whose distributions has a finite first moment. Let $l_*$ be the largest lower endpoint satisfying $\sup_QQ(B<l)\le\tau$, and $u_*$ the smallest upper endpoint satisfying $\sup_QQ(B>u)\le\tau$.
\begin{proposition}[Expected-score optimality]\label{main-prop:protected}
The finite extrema exist, $l_*\le u_*$, and $R^*=[l_*,u_*]$ minimizes expected interval score for every $Q\in\mathcal Q$ among intervals satisfying those same two tail requirements. If $\mathcal Q\subseteq\mathcal Q'$, every member of $\mathcal Q'$ is integrable and $R^*$ remains feasible there, the conclusion also holds for intervals satisfying these tail requirements for $\mathcal Q'$.
\end{proposition}
\end{samepage}

\subsection{Component accuracy}\label{main-sec:component-accuracy}

The preceding analysis holds the supplied component distributions fixed. Errors in those inputs can change both the probability under independence and its sharp dependence bound. We quantify how much component error can be tolerated while preserving the contrast between their conclusions. This addresses marginal accuracy alongside dependence uncertainty \citep{Bartl2022}.

Keep $m$ and the outcome map $g$ fixed, and let $H=\{(e,d):g(e,d)\ge0\}$ be the event that inflation does not slow. For component probability laws $\mu,\nu$, define
\[
 p_I(\mu,\nu)=(\mu\otimes\nu)(H),\qquad
 U(\mu,\nu)=\sup_{Q\in\Pi(\mu,\nu)}Q(H).
\]
Write $F_E,F_D$ for component cumulative distribution functions (CDFs). Their Kolmogorov distance is $d_K(F,G)=\sup_x|F(x)-G(x)|$. Suppose the unknown current functions $F_E^*,F_D^*$ satisfy $d_K(F_E^*,\widehat F_E)\le r_E$ and $d_K(F_D^*,\widehat F_D)\le r_D$, and put $\rho=r_E+r_D$. Stars denote current probabilities, and hats probabilities under the supplied distributions. Then
\begin{equation}\label{main-eq:product-transport-modulus}
 |p_I^*-\widehat p_I|\le\rho,\qquad |U^*-\widehat U|\le\rho.
\end{equation}
Appendix~\ref{main-app:proof-margin} proves these bounds, retaining point masses and the stated event boundaries.

\begin{samepage}
\begin{corollary}[Dependence binding with imperfect margins]\label{main-cor:imperfect-margins}
If
\begin{equation}\label{main-eq:binding-allowance}
 \rho\le\gamma-\widehat p_I,\qquad \rho<\widehat U-\gamma,
\end{equation}
then for every admissible pair of current margins in the stated discrepancy bounds, its product coupling supports the $1-\gamma$ deceleration assertion, while at least one coupling of those same margins fails it. The failing coupling may differ across marginal pairs. All-couplings support instead persists if $\rho\le\gamma-\widehat U$.
\end{corollary}
\end{samepage}

The two conditions express the required accuracy of the component distributions relative to the decision threshold. The allowance $\rho$ must fit within the independence probability's margin below $\gamma$ and remain strictly smaller than the worst-case probability's margin above it. Under these conditions, independence continues to support the assertion for every admissible pair of current component distributions.

\subsection{Precision from additional revisions}\label{main-sec:rare-tail-learning}

The preceding bounds quantify how errors in the component distributions affect directional conclusions. We now isolate the information contributed by additional revision records to estimating one fixed candidate bound, keeping the initial-error estimator unchanged. We establish conditions under which these records improve precision uniformly over dependence within completed pairs.

Use the nondecreasing outcome map $g(e,d)=m_s+e+d$, and keep the origin state and an $E$ threshold $x$ fixed independently of the observations used for estimation. Under the fixed target component laws, let $p_E=P(E>x)$ and $p_D=P\{g(x,D)\ge0\}$. Then
\[
 \{g(E,D)\ge0\}\subseteq\{E>x\}\cup\{g(x,D)\ge0\}.
\]
The sufficient bound is $U_x=p_E+p_D$. To estimate it, define historical indicators $X_i=1\{E_i>x\}$ and $Y_i=1\{g(x,D_i)\ge0\}$. Revision-only observations contribute directly to the estimate of $p_D$ while the initial-error estimator is held fixed.

Suppose there are $n\ge1$ completed pairs, $n_E\ge n$ first-release errors and $N>n$ revisions, with every pair retained and indexed $i=1,\ldots,n$ in both histories. Observations from distinct dates, including early-only and revision-only observations, are independent for this sampling result. Assume that each historical $E$ indicator has target mean $p_E$ and each historical $D$ indicator has target mean $p_D$. The indicators $X_i$ and $Y_i$ within pair $i$ may have covariance $c_i$, different across pairs. Historical-to-target alignment is imposed at these two tail events. Keep fixed nonnegative $E$ weights $w_i$ summing to one, and give revisions equal weights. Compare
\[
 \widehat U_P=\sum_{i=1}^{n_E}w_iX_i+\frac1n\sum_{i=1}^nY_i,
 \qquad
 \widehat U_F=\sum_{i=1}^{n_E}w_iX_i+\frac1N\sum_{i=1}^NY_i.
\]
Let $v_X=p_E(1-p_E)$, $v_Y=p_D(1-p_D)$, $C_w=\sum_{i=1}^nw_ic_i$ and $w_P=\sum_{i=1}^nw_i\le1$.
\begin{proposition}[Sharp rare-tail information gain]\label{main-prop:rare-tail}
Both estimators are unbiased for $U_x$. Their variance difference is
\begin{equation}\label{main-eq:variance-gain}
 \Delta V=\operatorname{Var}(\widehat U_P)-\operatorname{Var}(\widehat U_F)
 =\left(\frac1n-\frac1N\right)(v_Y+2C_w).
\end{equation}
Over permissible within-pair indicator dependence, its sharp infimum is
\[
 \inf\Delta V=\left(\frac1n-\frac1N\right)
 [p_D(1-p_D)-2w_P\min\{p_E p_D,(1-p_E)(1-p_D)\}].
\]
If $p_D>0$ and $p_E+p_D\le\gamma<1/2$, enriching revisions strictly reduces variance and MSE uniformly over that dependence, even when $n_E=n$.
\end{proposition}

Additional revision records estimate the revision-tail probability $p_D$ more precisely. When the paired tail indicators have negative covariance, sampling fluctuations in the two component estimates partly offset one another. Adding revisions reduces the weight on the paired revisions and weakens this offset. The rare-tail condition keeps the lost offset small enough for the variance and MSE of the fixed-bound estimate to fall under every admissible dependence pattern within completed pairs. The initial-error estimate has the same sampling variance in both procedures. Further dependence results, extensions of the precision comparison and finite-support implementation details are in Sections~\ref{oa-sec:oa-confidence} and \ref{oa-sec:numerical} of the Online Appendix.

Related incomplete-record and data-fusion work studies information gains under restrictions linking sources (\citealp{Berrett2025}, and \citealp{Li2023}, \href{https://doi.org/10.1093/biomet/asaf005}{corrected 2025}). Here we examine the precision gained by adding revision observations to a fixed tail-probability bound.

\section{Monte Carlo experiments}\label{main-sec:learning}

We use two experiments to examine the value of complementary histories. The first studies the predictive gain from additional revision records and its dependence on their relevance to the current target. The second studies how publication delays affect predictions based on full and completed-pair component histories. We evaluate nominal 90 percent intervals using $\IS$, Equation~(\ref{main-eq:score}), with $\alpha=0.10$ and $\tau=0.05$.
\label{main-sec:conditional-simulation}
We generate 128 independent paths, each with independent Gaussian draws $m_{s_j}\sim N(0,1)$ and $E_j,u_j\sim N(0,1/4)$. First releases satisfy $A_j=m_{s_j}+E_j$, revisions satisfy $D_j=\eta_{c_j}+0.5A_j+u_j$, and later releases are $B_j=A_j+D_j$. The month effect is $\eta_c=0.25\sin\{2\pi(c-1)/12\}$. Revisions become available after twelve months.

Each path supplies twelve related forecast origins. Full and matched procedures share the forecast, initial-error history, specification and outcomes. Adding 200 older $(A,D)$ records expands training from 50--61 to 250--261 observations. These records follow the current revision law in the stable case. A paired sensitivity case raises only their intercept by one percentage point. Comparisons include revision regression using completed forecast errors as predictors, related to \citet{Croushore2019}, and two references combining the true conditional revision law with either the estimated or true forecast-error distribution.

\begingroup\setlength{\tabcolsep}{4pt}
\begin{table}[htbp]\centering\small
\caption{Additional release histories: conditional prediction and transfer}\label{main-tab:conditional-transfer-main}
\begin{tabular}{lrrrr}
\toprule Method/case & $\IS$ & Coverage & Width & Mean MSE\\\midrule
Full: stable & $3.862$ & $0.886$ & $2.978$ & $0.843$\\
Full: historical shift & $5.020$ & $0.827$ & $3.297$ & $1.451$\\
Matched pairs & $4.008$ & $0.893$ & $3.113$ & $0.900$\\
Revision regression & $4.408$ & $0.906$ & $3.522$ & $1.107$\\
Known revision law & $3.822$ & $0.893$ & $2.950$ & $0.830$\\
Gaussian (True-DGP) & $3.791$ & $0.895$ & $2.965$ & $0.820$\\
\bottomrule\end{tabular}
\par\medskip
\begin{tabular}{lrrr}\toprule Score contrast & Mean & Standard error & 95\% interval\\\midrule
Full minus matched: stable & $-0.146$ & $0.039$ & $[-0.223,-0.069]$\\
Full minus matched: shift & $1.012$ & $0.11$ & $[0.803,1.222]$\\
Shift minus stable contrast & $1.158$ & $0.10$ & $[0.956,1.360]$\\
\bottomrule\end{tabular}

\notes{Standard errors use the 128 independent path averages, each averaging twelve related origins. Approximate 95 percent intervals use Student $t$ critical values with 127 degrees of freedom.}
\end{table}

\endgroup

Compatible older records improve prediction chiefly by improving placement. Mean interval score falls from 4.01 to 3.86, and predictive-mean MSE from 0.90 to 0.84. Almost all of that gain comes from the midpoint contribution, about $-0.14$, rather than the half-width contribution of $-0.0040$. Coverage is slightly lower. The full procedure approaches the reference that knows the conditional revision law but uses the same estimated error distribution. Table~\ref{main-tab:conditional-transfer-main} reports the paired contrasts and their path-level Monte Carlo intervals.

In the sensitivity case, the added older records follow a different conditional revision law from the one governing the current target, so pooling them with the matched records introduces prediction bias. Full-history score rises to 5.02 against the unchanged matched score of 4.01, a cost of $1.012$. Coverage falls to 0.83 despite wider intervals, and the predicted mean moves about 0.80 percentage point above the same-error-input reference. Together, the two cases show why the older records help: they supply information about the current revision relationship, rather than a benefit from sample size alone.


A complementary simulation examines release timing for procedures that estimate both component distributions from completed error--revision pairs. It compares publication schedules on the same outcomes across 500 Gaussian paths with serial and annual dependence. When only initial-error volatility rises, earlier revisions reduce paired procedures' interval scores by 0.26--1.60, while full-history changes are below 0.02 in absolute value. Earlier publication restores access to recent initial errors already used by full histories. In a case with higher component volatility and changed dependence, the equal-weighted full-history dependence-robust interval attains 82--86 percent overall coverage at nominal 90 percent, highlighting the importance of component accuracy. Simulation designs and supporting results for the conditional-transfer and release-timing experiments are in Section~\ref{oa-sec:oa-simulations} of the Online Appendix.

\section{Empirical application}\label{main-sec:application}

The application uses headline and core CPI and PCE inflation in the United States to evaluate later-release predictions and the economic conclusions they support.

\subsection{Empirical design}\label{main-sec:scope}

Let $t$ denote the reference month and $v$ a published data vintage. For the seasonally adjusted price-index level $P_t(v)$, measure the change in inflation by
\begin{equation}\label{main-eq:momentum}
 M_t(v)=100\left[\left\{P_t(v)/P_{t-3}(v)\right\}^{4}
                  -\left\{P_{t-3}(v)/P_{t-6}(v)\right\}^{4}\right].
\end{equation}
This is the difference between successive three-month annualized inflation rates, measured in percentage points. A negative value means that inflation slows. Let $s$ be the forecast date for reference month $t$. Write $P_{r\mid s}$ for the price-index level for month $r$ in the common published vintage available at forecast date $s$. The monthly Cleveland nowcast $f_s$, in percent, gives a hypothetical level $\widehat P_t=P_{t-1\mid s}(1+f_s/100)$. Combining it with the available $P_{t-3\mid s}$ and $P_{t-6\mid s}$ in Equation~(\ref{main-eq:momentum}) gives the point forecast $m_s$.

Here $A_t$ is the change in inflation calculated from the first release and $B_t$ its value at the first annual update in the following calendar year. For CPI, the latter is the following February seasonal update.\footnote{\citet{Waller2024} discusses how this update can change the interpretation of inflation: seasonal revisions worsened the picture of 2022 but preserved the improvement in 2023.} For PCE, it is the following year's certified annual or comprehensive update. The lag therefore varies with the publication schedule. Each calculation uses one common published vintage for its three constituent price levels.

Predictions refer to the specified official publication. A latent economic quantity can be studied by adding a measurement model, as in \citet{Jacobs2011}. In this application, $\mathcal I_s$ contains the forecast, price levels available on its date, publication metadata and all eligible historical observations.

Interval evaluation covers reference months in 2022--2025 and uses 134 later-release outcomes, with 46 each for headline and core CPI and 21 each for headline and core PCE. These are the cases with the forecast-date inputs, first release and specified later annual publication needed for a common comparison. A further 20 observations have known initial errors but await the annual update and contribute only to the error history. PCE's shorter evaluation sample reflects missing forecast-date vintages: annual outcomes exist for all 48 months in 2022--2025, but 27 lack the coherent vintage needed to reconstruct $m_s$.

PCE is revision-rich, with roughly five revisions for each completed pair. CPI is much closer to matched: its median histories contain about 124 errors, 132 revisions and 113 pairs. The newest error is typically six reference months more recent than the newest pair for CPI and twelve for PCE. PCE revision records extend back to July 2000, whereas forecast-linked errors begin in September 2013.

The specifications and weighting rules are fixed across forecast dates. Each fit uses only records published by its forecast date. We use two weighting rules for the initial-error distribution: \emph{equal weighting} and \emph{recency weighting}. Equal weighting assigns the same weight to every eligible first-release error. Recency weighting assigns weights proportional to $2^{-(t-j)/12}$, halving weight with twelve reference months of age.\footnote{Exponentially weighted empirical quantiles have an established forecasting precedent in \citet*{Boudoukh1998}. We use a fixed twelve-month half-life.} Both retain every positive-weight record and leave revision weights and fits unchanged.

Interval scores follow Equation~(\ref{main-eq:score}) at $\alpha=0.10$, as in the simulations. The monthly comparisons are related: successive inflation changes use overlapping price levels, annual updates revise several months together, and headline and core inflation share underlying prices. We therefore group monthly score differences by year and use the four annual groups to summarize their variation, treating the series and both weighting rules jointly.

We compare the linear and spline conditional procedures with methods that use different parts of the available history. The independence procedure combines the full error and revision distributions under independence. The dependence-robust prediction interval meets its tail probability requirements under every dependence pattern with those same component distributions.

We include several additional comparison methods that augment the same supplied forecast. Direct calibration adds empirical quantiles of past total errors $E+D$ to that forecast, a residual-based approach related to conformal prediction \citep{Lei2018,Chernozhukov2021}. Revision regression uses completed error--revision pairs to regress revisions on initial errors, then averages the resulting Student predictive distributions over all eligible initial errors. Both require historical forecasts for their training observations. The historical-rank procedure combines the full error and revision distributions using the rank association in completed pairs, yielding directional probabilities and prediction intervals. Calibrated width scales the independence interval's width by a fixed factor estimated before evaluation. An equal-width diagnostic assesses how width and placement contribute to interval coverage.

We also evaluate the four-series joint SV model based on \citet{Clements2023} and the real-time-vintage benchmark based on \citet{Clements2017}. Both generate their own later-release predictions from the available release data.

Table~\ref{main-tab:practical} measures the value of additional release records within each conditional specification. Table~\ref{main-tab:benchmark-practical} assesses interval performance against forecast-linked revision regression and the two literature benchmarks, with independence as a component-based reference. Table~\ref{main-tab:directional} examines directional confidence under alternative dependence assumptions, holding the component distributions fixed for independence, historical-rank probability and the sharp bound. Direct calibration supplies the historical total-error alternative. Table~\ref{main-tab:pooled-protection} compares the coverage and width of dependence-robust intervals with independence and joint SV. Sections~\ref{oa-sec:oa-application}--\ref{oa-sec:numerical} of the Online Appendix give method and data specifications, source checks, complete series and calendar comparisons, and component-accuracy and coverage diagnostics.

\subsection{Additional release information}\label{main-sec:information}

The results in Table~\ref{main-tab:practical} show that full release histories lower mean interval scores across all four inflation series for both linear and spline conditional revision means. The gains are largest in PCE, where revision histories are richest. For the linear specification under equal weighting, scores fall from 4.045 to 3.014 for headline PCE and from 5.029 to 2.667 for core PCE.

\begin{table}[htbp]
\centering\small
\caption{Mean interval scores with full and matched release histories}\label{main-tab:practical}
\begin{tabular}{lrrrr}
\toprule
 & \multicolumn{2}{c}{\shortstack{Linear conditional\\revision mean}} & \multicolumn{2}{c}{\shortstack{Spline conditional\\revision mean}}\\
\cmidrule(lr){2-3}\cmidrule(lr){4-5}
Series & Full & Matched & Full & Matched\\
\midrule
\multicolumn{5}{l}{Equal weighting}\\
CPI & 4.134 & 4.331 & 4.020 & 4.197 \\
Core CPI & 3.080 & 3.131 & 3.186 & 3.240 \\
PCE & 3.014 & 4.045 & 2.919 & 3.753 \\
Core PCE & 2.667 & 5.029 & 2.768 & 5.600 \\
\midrule
\multicolumn{5}{l}{Recency weighting}\\
CPI & 4.027 & 4.253 & 3.886 & 4.095 \\
Core CPI & 3.188 & 3.240 & 3.241 & 3.289 \\
PCE & 2.895 & 4.151 & 2.842 & 3.813 \\
Core PCE & 2.881 & 4.690 & 2.937 & 5.224 \\
\bottomrule
\end{tabular}
\notes{Full histories use all eligible release records. Matched histories use the forecast-linked subset. Within each specification, both fits share the forecast and initial-error input. Equal and recency weighting refer to the initial-error inputs. The same 46/46/21/21 outcomes are used for CPI/core CPI/PCE/core PCE.}
\end{table}

Compared with the linear procedure, the spline lowers interval scores for headline CPI and PCE but modestly raises them for both core series.

In PCE, the full-history gains arise chiefly from avoiding a few large misses, even though matched histories perform better on more individual dates.

Full histories also lower the linear procedure's mean interval score in all 32 leave-one-year-out comparisons.

Better placement drives most of the PCE gain (Figure~\ref{main-fig:placement}). At matched width, moving to the full-history midpoint lowers equal-weighted score from 4.0 to 3.1 for headline PCE and from 5.0 to 3.1 for core PCE. The decomposition averages the two orders of changing midpoint and width, so the contributions sum to the score difference. Position improves score at either width in all eight series--weighting comparisons.

\begin{figure}[htbp]
\centering
\begin{tikzpicture}
\begin{axis}[width=.9\linewidth,height=4.7cm,xlabel={Contribution to full-minus-matched interval score},ytick={1,2,3,4,5,6,7,8},yticklabels={CPI O,CPI W,Core CPI O,Core CPI W,PCE O,PCE W,Core PCE O,Core PCE W},y dir=reverse,legend style={at={(.5,1.02)},anchor=south,legend columns=2},xmajorgrids=true]
\addplot+[only marks,mark=*,black] coordinates {(-.1698,1)(-.1895,2)(-.0561,3)(-.0331,4)(-.9262,5)(-1.2034,6)(-1.7847,7)(-1.4729,8)};
\addplot+[only marks,mark=square*,gray] coordinates {(-.0270,1)(-.0360,2)(.0051,3)(-.0194,4)(-.1042,5)(-.0523,6)(-.5777,7)(-.3361,8)};
\legend{Position,Width}
\end{axis}
\end{tikzpicture}
\caption{Decomposition of interval-score differences into position and width}\label{main-fig:placement}
\notes{Negative contributions favour full histories. Position and width contributions are in interval-score units; O/W denote equal/recency weights.}
\end{figure}


\FloatBarrier

\subsection{Forecasting benchmarks}\label{main-sec:benchmark-practical}\label{main-sec:benchmarks}

We compare the interval performance of the full-history procedures with the forecasting benchmarks from the literature. Table~\ref{main-tab:benchmark-practical} shows that both procedures have lower mean interval scores than joint SV in every series. In the pooled sample, score reductions are approximately 27--28 percent. With recency weighting, both procedures match joint SV's coverage of 124/134 outcomes using intervals approximately 34--36 percent narrower, with three severe misses compared with four.

\begin{table}[htbp]
\centering\small\setlength{\tabcolsep}{3pt}
\caption{Full-history conditional predictions and forecasting benchmarks}\label{main-tab:benchmark-practical}
\begin{tabular}{lrrrrrrrr}
\toprule
& \multicolumn{4}{c}{CPI ($n=46$)} & \multicolumn{4}{c}{Core CPI ($n=46$)}\\
\cmidrule(lr){2-5}\cmidrule(lr){6-9}Method & $\IS$ & Cov. & Width & Sev. & $\IS$ & Cov. & Width & Sev.\\
\midrule
Linear conditional & 4.134 & 42 & 2.661 & 4 & 3.080 & 40 & 1.933 & 2 \\
Spline conditional & 4.020 & 42 & 2.586 & 4 & 3.186 & 40 & 1.868 & 2 \\
Independence & 4.424 & 43 & 3.256 & 3 & 2.923 & 44 & 2.367 & 1 \\
Revision regression & 4.292 & 42 & 2.933 & 3 & 2.779 & 40 & 1.869 & 1 \\
Joint SV & 6.217 & 43 & 5.278 & 2 & 3.442 & 42 & 2.808 & 2 \\
Real-time vintage & 6.341 & 44 & 5.039 & 2 & 3.398 & 42 & 2.945 & 0 \\
\addlinespace\multicolumn{1}{l}{} & \multicolumn{4}{c}{PCE ($n=21$)} & \multicolumn{4}{c}{Core PCE ($n=21$)}\\
\cmidrule(lr){2-5}\cmidrule(lr){6-9}Method & $\IS$ & Cov. & Width & Sev. & $\IS$ & Cov. & Width & Sev.\\
Linear conditional & 3.014 & 18 & 2.622 & 0 & 2.667 & 18 & 2.141 & 0 \\
Spline conditional & 2.919 & 19 & 2.540 & 0 & 2.768 & 17 & 2.126 & 0 \\
Independence & 3.196 & 19 & 2.771 & 0 & 2.867 & 19 & 2.203 & 1 \\
Revision regression & 3.257 & 18 & 2.301 & 1 & 3.163 & 16 & 1.767 & 1 \\
Joint SV & 4.413 & 20 & 4.384 & 0 & 4.065 & 19 & 3.466 & 0 \\
Real-time vintage & 3.878 & 19 & 3.504 & 0 & 5.708 & 12 & 2.341 & 3 \\
\bottomrule
\end{tabular}
\notes{Both conditional procedures use full release histories. We use equal weighting for the initial-error inputs. Cov. counts covered outcomes. Sev. counts misses more than 0.5 percentage point outside the interval. Width and score are percentage points. Table~\ref{main-tab:practical} isolates the full-versus-matched history comparison.}
\end{table}

For core PCE under equal weighting, the linear procedure covers 18/21 outcomes with no severe misses, compared with 12/21 and three severe misses for the real-time-vintage benchmark. Its mean interval score is less than half as large, and its intervals are narrower.

Both conditional procedures also lower interval scores relative to revision regression in headline CPI and both PCE series. Revision regression gives the lowest core-CPI score. Under equal weighting, joint SV covers six more pooled outcomes than the linear procedure, with wider intervals and smaller mean outside distance.

\FloatBarrier


\subsection{Directional confidence}\label{main-sec:june}

The dependence bound identifies which high-confidence deceleration assertions rely on an assumed association between initial errors and revisions. We assess one-sided requirements of 90 and 95 percent probability that $B_t<0$, holding the component distributions fixed. The last column of Table~\ref{main-tab:directional} shows that, at the 95 percent requirement, the dependence bound withholds every incorrect deceleration assertion made under independence.

\begin{table}[htbp]
\centering\small\setlength{\tabcolsep}{3pt}
\caption{Deceleration assertions and later-release outcomes}\label{main-tab:directional}
\begin{tabular}{llrrrrr}
\toprule
Target & Weights & Independence & \shortstack{Historical-rank\\probability} & \shortstack{Direct\\calibration} & \shortstack{Dependence\\bound} & \shortstack{Incorrect\\assertions\\withheld} \\
\midrule
\multicolumn{7}{l}{95 percent requirement}\\
CPI & Equal & 11/1 & 13/1 & 13/1 & 8/0 & 1 \\
CPI & Recency & 11/0 & 12/0 & 13/1 & 7/0 & 0 \\
Core CPI & Equal & 11/0 & 13/0 & 12/0 & 5/0 & 0 \\
Core CPI & Recency & 7/0 & 10/0 & 12/0 & 3/0 & 0 \\
PCE & Equal & 6/0 & 6/0 & 7/0 & 3/0 & 0 \\
PCE & Recency & 5/0 & 7/0 & 7/0 & 3/0 & 0 \\
Core PCE & Equal & 3/1 & 2/1 & 3/1 & 1/0 & 1 \\
Core PCE & Recency & 2/1 & 2/1 & 3/1 & 1/0 & 1 \\
\midrule
\multicolumn{7}{l}{90 percent requirement}\\
CPI & Equal & 13/1 & 13/1 & 14/1 & 8/0 & 1 \\
CPI & Recency & 13/0 & 13/0 & 14/1 & 9/0 & 0 \\
Core CPI & Equal & 14/0 & 16/0 & 16/0 & 9/0 & 0 \\
Core CPI & Recency & 11/0 & 13/0 & 16/0 & 9/0 & 0 \\
PCE & Equal & 8/0 & 8/0 & 8/0 & 4/0 & 0 \\
PCE & Recency & 7/0 & 7/0 & 8/0 & 3/0 & 0 \\
Core PCE & Equal & 3/1 & 3/1 & 6/1 & 1/1 & 0 \\
Core PCE & Recency & 3/1 & 3/1 & 6/1 & 1/1 & 0 \\
\bottomrule
\end{tabular}
\notes{The four procedure columns report correct/incorrect deceleration assertions against the later release. The final column counts incorrect independence assertions withheld by the dependence bound, matching the same target, date and weighting rule. Historical-rank probability uses reconstructed historical dependence. Direct calibration uses the empirical distribution of historical total errors $E+D$. Samples contain 46/46/21/21 dates for CPI/core CPI/PCE/core PCE.}
\end{table}

For example, the supplied forecast for June 2025 core PCE, $m=-1.56$ percentage points, suggests deceleration. Under equal weighting, independence assigns about 98.4 percent probability to deceleration, and the sharp lower probability across compatible dependence patterns is about 91.5 percent. At the 95 percent requirement, the dependence bound withholds the assertion supported by independence and direct calibration. At 90 percent it retains a deceleration assertion. The later release records acceleration, $B=+0.12$, contradicting that retained assertion. The dependence bound thus screens out an incorrect high-confidence assertion by requiring support across compatible dependence patterns.

This robustness also has a cost in correct directional assertions. For equal-weighted core CPI, the requirement retains five of eleven 95 percent independence assertions. All six withheld conclusions are confirmed by later outcomes.

\FloatBarrier
The June confidence comparison also depends on the accuracy of the component distributions. Corollary~\ref{main-cor:imperfect-margins} quantifies sufficient accuracy for independence to continue supporting June's 95 percent assertion and for the dependence bound to continue withholding it. The sufficient limits on the sum of the two component CDF errors are about 3.4 probability percentage points with equal weights and 1.8 with recency weights.\footnote{Each CDF error is the supremum, over thresholds, of the absolute difference between a supplied component CDF and its current counterpart. Discrete working laws can make small accuracy allowances demanding. An atom of mass $p$ is at Kolmogorov distance at least $p/2$ from every continuous CDF. June's largest recency-weighted atom gives a bound of about 5.8 probability percentage points.} June's 95 percent contrast holds under both weighting rules if the two component CDF errors sum to at most one percentage point.

\subsection{Dependence-robust prediction intervals}\label{main-sec:empirical-protection}

A prediction interval provides a two-sided range for the later publication. Its performance depends on coverage, width and the size of misses. Our 90 percent dependence-robust intervals allocate a 5 percent probability budget to each outside tail under every dependence pattern with the supplied component distributions, as in Equation~(\ref{main-eq:robust-tails}).

\begin{table}[htbp]
\centering\small\setlength{\tabcolsep}{5pt}
\caption{Pooled interval performance across the four inflation series}\label{main-tab:pooled-protection}
\begin{tabular}{lrrrr}
\toprule
Method & \shortstack{Covered\\out of 134} & \shortstack{Severe\\misses} & Width & $\IS$\\
\midrule
\multicolumn{5}{l}{Equal weighting}\\
Independence & 125 & 5 & 2.710 & 3.472 \\
Joint SV & 124 & 4 & 4.006 & 4.644 \\
Dependence-robust & 131 & 1 & 4.466 & 4.624 \\
\midrule
\multicolumn{5}{l}{Recency weighting}\\
Independence & 126 & 4 & 3.093 & 3.674 \\
Joint SV & 124 & 4 & 4.006 & 4.644 \\
Dependence-robust & 132 & 0 & 4.666 & 4.757 \\
\bottomrule
\end{tabular}
\notes{The sample pools 46/46/21/21 outcomes for CPI/core CPI/PCE/core PCE. Joint SV forecasts repeat across panels. Severe misses exceed 0.5 percentage point.}
\end{table}

\FloatBarrier

Across the 134 series--month outcomes, dependence-robust intervals cover seven to eight more outcomes than joint SV, with at most one severe miss (Table~\ref{main-tab:pooled-protection}). Their mean widths are 11.5--16.5 percent greater, with comparable mean interval scores. Relative to independence, they cover six additional outcomes under each rule, with wider intervals and higher mean scores.\footnote{An interval below zero declares deceleration, one above zero declares acceleration, and one containing zero leaves the direction open. The wider dependence-robust intervals produce fewer directional declarations.}

For June 2025 core PCE, the dependence-robust interval contains zero and covers the later acceleration, leaving the direction open. The independence and conditional intervals lie below zero and miss the outcome. A 90 percent dependence-robust interval lying strictly below zero would support at least 95 percent probability of deceleration across all compatible associations. Its upper-tail requirement explains why the 90 percent one-sided assertion in Section~\ref{main-sec:june} coexists with an interval containing zero.

An equal-width diagnostic keeps the independence interval's midpoint and borrows the dependence-robust prediction interval's width at each date. It matches the dependence-robust prediction interval's coverage and score in both PCE series and closely reproduces the CPI results. Most of the additional coverage therefore comes from wider intervals. The dependence-robust construction supplies that width through explicit tail requirements.

\FloatBarrier

\section{Conclusion}

An institution can use a longer official-release history to improve a forecast it already has. Our framework separates the information that helps predict a later publication from the dependence assumptions needed to support an economic conclusion. Its practical use depends on the relevance of past revisions and the robustness the decision requires. Future work could develop data-driven assessments of that relevance and extend the conditional relationship to evolving revision processes. A further challenge is to learn dependence restrictions from delayed releases while accounting for uncertainty in the component distributions.

\section*{Acknowledgements}
I thank Raffaella Giacomini (the editor), two anonymous referees, Kristoffer Nimark, Ekaterina Kazak, Zebang Xu and seminar participants for helpful comments and suggestions. Any errors are my own.

\appendix
\begingroup\postdisplaypenalty=10000
\section{Proofs}\label{main-app:proofs}

\label{main-app:proof-bridge}

\begin{proof}[Proof of Proposition~\ref{main-prop:bridge}]
Let $\mu^H$ be a probability mixture assigning positive weight to each first-release law from eligible historical joint $(A,D)$ records with $c_j=c_t$. Historical alignment identifies $G(\cdot\mid a,c_t)$ for $\mu^H$-almost every $a$. By the domination assumption, the current first-release law is absolutely continuous with respect to $\mu^H$, so this identifies the version required for current prediction. Current conditional transfer and the separately identified initial-error distribution give the joint law of $(E_t,D_t)$ conditional on $\mathcal I_s$,
\[
 K_s(de,dd)=F_{E,s}(de)G(dd\mid m_s+e,c_t).
\]
Consequently,
\[
\begin{aligned}
 P(B_t\le b\mid\mathcal I_s)
 &=\int1\{m_s+e+d\le b\}\,K_s(de,dd)\\
 &=\int G(b-m_s-e\mid m_s+e,c_t)\,dF_{E,s}(e).
\end{aligned}
\]
\end{proof}

\label{main-app:proof-sharp}

\begin{proof}[Proof of Equation~(\ref{main-eq:sharp-probabilities})]
Let $e_1<\cdots<e_K$ be the distinct $E$ support points, with positive masses $\mu_i$. Partition the real line at the thresholds $-m_s-e_i$ into singleton threshold points and the intervening open intervals. Discard cells of zero $\widehat\nu$-mass. For each remaining cell $C_j$, choose a representative $d_j\in C_j$ and put $\nu_j=\widehat\nu(C_j)$. The event vector $(1\{g(e_i,d)\ge0\})_i$ is constant on each cell, including singleton cells at the thresholds. For $H_{ij}=1\{g(e_i,d_j)\ge0\}$ define $N_H(A)=\{j:H_{ij}=1\text{ for some }i\in A\}$. An event-flow network has source-row capacities $\mu_i$, unrestricted event-edge capacities and column-sink capacities $\nu_j$. Its maximum event mass equals
\[
 \min_{A\subseteq\{1,\ldots,K\}}
 \left[\sum_{i\notin A}\mu_i+\sum_{j\in N_H(A)}\nu_j\right].
\]
Increasing row neighbourhoods allow a nonempty A to be replaced by $\{1,\ldots,\max A\}$ without changing its neighbours or increasing cost. Thus the minimum is $\min_k U_k$, where
\[
 U_k=\widehat\mu(E>e_k)+\widehat\nu\{g(e_k,D)\ge0\}.
\]
The strict complement has decreasing neighbourhoods, so suffix cuts give $\min_k L_k$, with
\[
 L_k=\widehat\mu(E<e_k)+\widehat\nu\{g(e_k,D)<0\}.
\]
The empty cut costs one and is dominated by an extreme finite anchor.

For a maximizing flow f, set $r_i=\mu_i-\sum_jf_{ij}$, $c_j=\nu_j-\sum_if_{ij}$ and $R_0=\sum_ir_i=\sum_jc_j$. Positive $r_ic_j$ on an event edge would permit an augmentation, contradicting maximality. If $R_0>0$, $\pi_{ij}=f_{ij}+r_ic_j/R_0$ completes the marginals without extra event mass. At $R_0=0$ the flow is already complete. Lift the cell coupling to the original revision margin by assigning $\pi_{ij}\widehat\nu(dd\mid C_j)$ to $E=e_i$ and $D\in C_j$. Applying the construction to H and its complement gives the attained range
\[
 [\min_QQ(H),\max_QQ(H)]=[1-\min_kL_k,\min_kU_k].
\]
Convexity supplies every intermediate value. The events $g\ge0$ and $g<0$ partition the support. Different events need not have a common extremizing dependence pattern.
\end{proof}

\label{main-app:proof-protected}

\begin{proof}[Proof of Proposition~\ref{main-prop:protected}]
Set $\tau=\alpha/2$, $\mathcal L=\{l:\sup_QQ(B<l)\le\tau\}$ and $\mathcal U=\{u:\sup_QQ(B>u)\le\tau\}$. Uniform tightness makes the sets nonempty. Any fixed law in the nonempty class bounds their extrema, so $l_*\!=\sup\mathcal L$ and $u_*\!=\inf\mathcal U$ are finite. For admissible sequences $l_n\uparrow l_*$ and $u_n\downarrow u_*$, each Q gives
\[
 Q(B<l_*)=\lim_n Q(B<l_n)\le\tau,\qquad
 Q(B>u_*)=\lim_n Q(B>u_n)\le\tau.
\]
Both extrema are feasible. If $l_*>u_*$, $1\le Q(B<l_*)+Q(B>u_*)\le2\tau=\alpha<1$, a contradiction. Every $J=[l,u]$ satisfying these tail requirements contains $R^*=[l_*,u_*]$. With $d_l=l_*-l\ge0$, $d_u=u-u_*\ge0$,
\[
\begin{aligned}
 (l-B)_+-(l_*-B)_+&\ge-d_l1\{B<l_*\},\\
 (B-u)_+-(B-u_*)_+&\ge-d_u1\{B>u_*\}.
\end{aligned}
\]
Finite first moments make the expected scores finite, and
\[
 \mathbb E_Q[IS_\alpha(J;B)-IS_\alpha(R^*;B)]
 \ge d_l[1-Q(B<l_*)/\tau]+d_u[1-Q(B>u_*)/\tau]\ge0.
\]
If $\mathcal Q\subseteq\mathcal Q'$, every interval satisfying the tail requirements for $\mathcal Q'$ also satisfies them for $\mathcal Q$ and contains $R^*$. Retained feasibility of $R^*$ and integrability for every $\mathcal Q'$ member give the same inequality. Neither uniqueness nor realized-score dominance follows.
\end{proof}

\label{main-app:proof-margin}

\begin{proof}[Proof of Corollary~\ref{main-cor:imperfect-margins}]
For measurable monotone $0\le f\le1$, layer-cake integration gives
\[
 \left|\int f\,d\mu-\int f\,d\mu^*\right|\le d_K(\mu,\mu^*),
\]
because the level sets of f are open or closed ordered rays. Left limits handle atoms. Apply this successively to conditional probabilities of the fixed H to obtain $|p_I^*-\widehat p_I|\le\rho$.

For a coordinatewise decreasing event $A$, let $T(F,G)$ be either extremal event probability over all couplings. If $d_K(F,F^*)\le r$, form at finite $x$ the brackets
\[
 F^\downarrow(x)=(F(x)-r)_+,\qquad
 F^\uparrow(x)=\min\{F(x)+r,1\},
\]
placing displaced mass at formal ordered endpoints. Monotonicity gives
\[
 T(F^\downarrow,G)\le T(F^*,G)\le T(F^\uparrow,G).
\]
Each bracket can be coupled to F with disagreement probability at most r. Gluing that coupling to a joint law with other margin G changes event probability by at most r. Applying this in both directions bounds each bracket extremum within r of $T(F,G)$. Repeating one coordinate at a time yields
\[
 |T(F_E^*,F_D^*)-T(F_E,F_D)|\le r_E+r_D.
\]
The formal endpoint sections are defined by one-sided limits. This argument does not infer total variation proximity between arbitrary F and $F^*$. Apply it to $A=H^c$ and $U(H)=1-\inf_QQ(H^c)$ to obtain $|U^*-\widehat U|\le\rho$. Keep the forecast and event fixed as the component distributions vary. Here $H=\{g\ge0\}$ and $H^c=\{g<0\}$.

For every admissible current distribution pair,
\[
 p_I^*\le\widehat p_I+\rho\le\gamma<\widehat U-\rho\le U^*.
\]
A supremum above $\gamma$ implies some violating joint distribution, without requiring attainment or the same joint distribution for different component pairs. Likewise $U^*\le\widehat U+\rho\le\gamma$ proves retention under all dependence patterns. The sufficient binding set is
\[
 [0,\gamma-\widehat p_I]\cap[0,\widehat U-\gamma).
\]
Its upper endpoint is included only when the independence slack is strictly smaller. Equality qualifies on the independence side, but withdraws strict non-support on the unrestricted side. Empty and singleton sets remain distinct from unavailable inputs.
\end{proof}

\label{main-app:proof-rare}

\begin{proof}[Proof of Proposition~\ref{main-prop:rare-tail}]
Condition on the fixed origin state, threshold and deterministic nonnegative weights, including possible zero weights. Every completed pair remains. Put $S_E=\sum_iw_iX_i$, $P_n=n^{-1}\sum_{i\le n}Y_i$ and $F_N=N^{-1}\sum_{i\le N}Y_i$. The stated alignment and across-record independence give
\[
\begin{aligned}
 \mathbb E(S_E,P_n,F_N)&=(p_E,p_D,p_D),\qquad
 \operatorname{Var}(S_E)=v_X\sum_iw_i^2,\\
 \operatorname{Var}(P_n)&=v_Y/n,\quad \operatorname{Var}(F_N)=v_Y/N,\\
 \operatorname{Cov}(S_E,P_n)&=C_w/n,\quad
 \operatorname{Cov}(S_E,F_N)=C_w/N.
\end{aligned}
\]
Since $\widehat U_P=S_E+P_n$ and $\widehat U_F=S_E+F_N$,
\[
 \Delta V=(1/n-1/N)(v_Y+2C_w).
\]
For Bernoulli indicators, the least joint-success probability is $q_{\min}=\max(0,p_E+p_D-1)$, hence
\[
\begin{aligned}
 c_i&\ge q_{\min}-p_E p_D=-\min\{p_E p_D,(1-p_E)(1-p_D)\},\\
 C_w&\ge-w_P\min\{p_E p_D,(1-p_E)(1-p_D)\}.
\end{aligned}
\]
Independent pairs attaining $q_{\min}$, with independent aligned early-only and auxiliary units, attain this bound simultaneously. Under $p_D>0$ and $p_E+p_D\le\gamma<1/2$,
\[
 \inf\Delta V=(1/n-1/N)p_D(1-p_D-2p_E w_P)
 \ge(1/n-1/N)p_D(1-2\gamma+p_D)>0.
\]
Unbiasedness gives the same MSE ordering.
\end{proof}

\endgroup

\StartOnlineAppendix

\noindent
Section~\ref{oa-sec:oa-conditional} details the linear and spline conditional procedures.
Section~\ref{oa-sec:oa-confidence} provides theoretical examples and extensions for confidence bounds, robust intervals, component accuracy and precision from additional revisions.
Section~\ref{oa-sec:oa-simulations} presents the Monte Carlo designs and results for additional revision records and publication delays.
Section~\ref{oa-sec:oa-application} documents the inflation data, comparison methods and full forecasting results.
Section~\ref{oa-sec:oa-protection-evidence} reports directional assertions, sensitivity to component accuracy and interval performance in the application.
Section~\ref{oa-sec:numerical} gives the numerical conventions for finite-support calculations and predictive mixtures.

\FloatBarrier
\section{Conditional prediction}\label{oa-sec:oa-setup}\label{oa-sec:oa-conditional}

The first and later changes in inflation are $A,B$, the supplied forecast is $m_s$, the initial error is $E=A-m_s$, and the revision is $D=B-A$. Published first--later release pairs reveal $(A,D)$ without a historical forecast. A forecast-linked record also reveals $E$ and the total error $E+D$. Conditional prediction uses joint $(A,D)$ records, whereas the component analysis can also use observations containing $D$ alone.

Equal weighting (O) assigns the same mass to each eligible record. Recency weighting (W) changes only $E$, with mass proportional to $2^{-(t-j)/12}$ for reference-month age $t-j$. The revision distribution is unchanged. All positive masses are retained and normalized exactly. Section~\ref{oa-sec:operations} gives the quantile and tie conventions.

\subsection{Linear conditional prediction}\label{oa-sec:bridge}

Full and matched fits use all eligible joint release records and their forecast-linked subset, respectively, while sharing the forecast, full initial-error distribution and model specification. Their comparison measures how the additional records change the fitted prediction. The common-law, current-transfer and absolute-continuity conditions in the main text justify using the estimated revision relationship for the current target. The linear specification is
\[
 D=x(A,c)'\beta+u,\qquad
 x=(1,A,1\{c=2\},\ldots,1\{c=12\})'.
\]
Here $c$ indexes the reference month of the year, $\beta$ is the coefficient vector and $u$ the regression error. January is the omitted month category. Training records have equal weights. Let $X$ be the design matrix for the $n$ training observations, with rank $r$ and residual degrees of freedom $\nu=n-r$. The fit supplies $\widehat\beta$ and residual scale $\widehat\sigma$. For an estimable row $x$, the fitted conditional revision predictive CDF is
\[
 \widehat G_s(d\mid A,c)
 =T_\nu\left(\frac{d-x'\widehat\beta}
 {\widehat\sigma\sqrt{1+x'(X'X)^+x}}\right).
\]
Here $T_\nu$ is the standard Student CDF and $(X'X)^+$ is the Moore--Penrose inverse. Under Gaussian, homoskedastic, independent regression errors this is the conventional Student predictive CDF. For serially related release records it is a working predictive law under the transfer and absolute-continuity restrictions in the main text.

Let $\widehat F_{E,s}$ denote the initial-error CDF supplied or estimated at forecast date $s$. Before the first release, $A_t$ is unknown. Integrating the fitted revision law over possible initial errors gives the estimated later-release CDF
\begin{equation}\label{oa-eq:bridge}
 \widehat F_{B,s}(b)=\int
 \widehat G_s(b-m_s-e\mid m_s+e,c_t)\,d\widehat F_{E,s}(e).
\end{equation}
This finite Student mixture holds the $E$ input and model form fixed when changing the $(A,D)$ training records.

For June 2025 core PCE, the equal-weighted calculation starts from $m=-1.56$ percentage points and 62 stored initial errors. Each error implies a possible first release, $m+e_k$, at which the conditional revision law fitted to 270 historical $(A,D)$ records is evaluated. Shifting by that first release and averaging the 62 distributions gives the later-release prediction in Equation~(\ref{main-eq:bridge}). The matched fit uses 55 forecast-linked records with the same forecast and error input.

The fitted mixture requires positive residual degrees of freedom, a positive residual scale and an estimable prediction at every positive-mass hypothetical first release. All 134 evaluated outcomes meet these conditions. Predictions beyond the historical $A$ range rely on the fitted functional form. Section~\ref{oa-sec:source-diagnostics} quantifies this extrapolation.\par\medskip\noindent\textit{Predictive mean.}\label{oa-sec:prediction-mean}
For the linear constant-variance model, let $\widehat\eta_{s,c}$ combine the fitted intercept and month effect, and let $\widehat\beta_{A,s}$ be the first-release coefficient. With residual degrees of freedom $\nu>1$ and finite first moments of component locations and scales, the mean of the fitted predictive distribution is
\begin{equation}\label{oa-eq:bridge-mean}
 \int b\,d\widehat F_{B,s}(b)=\widehat\eta_{s,c_t}
 +(1+\widehat\beta_{A,s})(m_s+\bar E_s),\qquad
 \bar E_s=\int e\,d\widehat F_{E,s}(e).
\end{equation}
The fitted mixture's mean is the weighted mean of its component centres. Interval midpoints and medians describe other features of the distribution.

\subsection{Spline specification}\label{oa-sec:richer-specification}

The conditional procedures use linear or spline revision means with constant innovation variance. Within each specification, full and matched histories retain the same forecast and $E$ input.

For each series, let $k_1<k_2<k_3<k_4$ be generalized-inverse empirical A quantiles at $0.10,0.35,0.65,0.90$ using forecast-linked release records whose annual revision was published before 1 January 2022. These knots stay fixed at all subsequent origins and are shared by full and matched histories. The spline appends to the baseline design the two columns
\[
 q_j(a)=\frac{(a-k_j)_+^3
  -\frac{k_4-k_j}{k_4-k_3}(a-k_3)_+^3
  +\frac{k_3-k_j}{k_4-k_3}(a-k_4)_+^3}{(k_4-k_1)^2},
 \qquad j=1,2.
\]
This basis has linear tails and adds two coefficients. Its construction uses only shared pre-2022 matched A values, not the additional full-history revisions or the comparison outcomes.

Both specifications use the Student predictive law in Section~\ref{oa-sec:bridge}.

\FloatBarrier
\section{Confidence and robust intervals}\label{oa-sec:proofs}\label{oa-sec:oa-confidence}

Notation follows the main text, with $g(e,d)=m_s+e+d$. Non-deceleration is $g(E,D)\ge0$ and deceleration is $g(E,D)<0$.

\subsection{Dependence and transfer}

The current-transfer restriction connects historical information to the forecast target. Consider current E taking $-1,1$ equally often and two regimes $D=-E$ and $D=E$. They have the same component distributions but $E+D=0$ surely versus $-2,2$ equally often. Give both regimes the same \emph{joint observable law} of old pairs, publication metadata, recent revealed E and older revision-only information. Before recent revisions mature, all $\mathcal I_s$-measurable statistics then have the same sampling law, while $P(E+D>1)$ is zero versus one half. Observational equivalence, not equality of margins alone, prevents discrimination.

If the current outcome is entirely unrestricted, no almost-surely finite interval can guarantee positive worst-case coverage. Let $\mathcal H$ denote the observed history and $\xi$ any analyst randomization. Consider a prediction procedure that returns an interval $J(\mathcal H,\xi)=[l,u]$, with finite endpoints almost surely. Hold the joint law of $(\mathcal H,\xi)$ fixed and set $B=M$ for arbitrarily large constants. Under the resulting law $P_M$,
\[
 P_M\{B\in J\}\le P\{u\ge M\}\longrightarrow0.
\]
Dominated convergence applies to finite $u$. The adversary need not observe $\xi$. This argument does not apply to maintained restricted classes or procedures sometimes returning the whole real line.

\subsection{Robust interval criteria}

Interval scores use Equation~(\ref{main-eq:score}) in the main text, with $\tau=\alpha/2$. We write $\IS$ for the 90 percent case $\alpha=0.10$. In the examples below, $I$ denotes the independence interval and $R$ the dependence-robust interval for the same supplied component laws.

\par\medskip\noindent\textit{Gaussian amplification.}\label{oa-sec:gaussian-protection}
The Gaussian example in Section~\ref{main-sec:confidence} of the main text admits the following shock construction.

A revision can remove an early measurement disturbance or amplify an economic surprise that the first release understated. At a fixed forecast $m$, let
\[
 A=m+V+N,\qquad B=m+(1+\kappa)V+J_\kappa,\qquad 0\le\kappa\le1,
\]
where $V,N,J_\kappa$ are independent centred Gaussian variables with variances $0.9$, $0.1$ and $0.9(1-\kappa^2)$. The later release removes measurement noise $N$ and amplifies the common surprise through $\kappa V$. The resulting component and conditional laws are stated in Equation~(\ref{main-eq:amplification}).

Let $z_p$ be the standard normal $p$ quantile.

\begin{proof}[Derivation for Example~\ref{main-ex:amplification}]
\phantomsection\label{oa-app:proof-gaussian}
The identities $E=V+N$ and $D=\kappa V+J_\kappa-N$ give unit component variances and covariance $0.9\kappa-0.1$. Gaussian conditioning gives Equation (\ref{main-eq:amplification}). With $s_\kappa^2=1.8(1+\kappa)$, an interval $[m-w,m+w]$ has coverage $2\Phi(w/s_\kappa)-1$ and expected score
\[
 2w+40\{s_\kappa\varphi(w/s_\kappa)
                  -w\overline\Phi(w/s_\kappa)\},
\]
where $\Phi$ and $\varphi$ are the standard normal CDF and density, and $\overline\Phi=1-\Phi$. For any coupling of standard normal $E,D$,
\[
 P(E+D>2z_{0.975})\le P(E>z_{0.975})+P(D>z_{0.975})=0.05.
\]
The reflected inequality gives the lower tail. For $m<0$,
\[
 P(B\ge0)\le P(E\ge-m/2)+P(D\ge-m/2)
       =2\overline\Phi(-m/2).
\]
Hence $P(B<0)\ge2\Phi(-m/2)-1$. Substitution yields the reported numbers. If instead $A=m+V+N$ and $B=m+V+J$ with independent Gaussian variables, $\operatorname{Cov}(E,D)=-\operatorname{Var}(N)\le0$ and $\operatorname{Var}(E+D)\le\operatorname{Var}(E)+\operatorname{Var}(D)$. The Gaussian independence interval is conservative.
\end{proof}

A historical 90 percent interval has half-width $z_{0.95}\sqrt{1.8}$, and independence gives $z_{0.95}\sqrt2$. After the change from $\kappa=0$ to $\kappa=1$, coverage falls to about 76 and 78 percent respectively. Both intervals impose dependence restrictions that fail in the new regime, even though the conditional means remain linear.

For the two standard normal component laws, $m\pm2z_{0.975}$ meets both tail requirements in Equation~(\ref{main-eq:robust-tails}) of the main text under every dependence pattern. A narrower Gaussian envelope, $m\pm z_{0.95}\sqrt{3.6}$, meets them within the example's family $0\le\kappa\le1$. Table~\ref{oa-tab:amplification} reports both choices alongside the historical and independence forecasts.

\begin{table}[htbp]
\centering\small
\caption{Prediction intervals under Gaussian and unrestricted dependence}\label{oa-tab:amplification}
\begin{tabular}{lrrrrr}
\toprule
 & Half-width & \multicolumn{2}{c}{Coverage} & \multicolumn{2}{c}{Expected score}\\
Method & & $\kappa=0$ & $\kappa=1$ & $\kappa=0$ & $\kappa=1$\\
\midrule
Fixed $\kappa=0$ law & 2.207 & 0.900 & 0.755 & 5.535 & 9.004\\
Independence & 2.326 & 0.917 & 0.780 & 5.556 & 8.688\\
Gaussian envelope & 3.121 & 0.980 & 0.900 & 6.424 & 7.827\\
Dependence-robust & 3.920 & 0.997 & 0.961 & 7.867 & 8.379\\
\bottomrule
\end{tabular}
\notes{All intervals centre on $m$. Coverage and expected score are calculated from the stated Gaussian laws.}
\end{table}

The dependence-robust prediction interval retains about 96 percent coverage after the change from $\kappa=0$ to $\kappa=1$. The Gaussian envelope achieves 90 percent with less width and lower expected score because its restriction is correct here. The difference in width reflects the class of admitted dependence patterns. Proposition~\ref{main-prop:protected} in the main text establishes expected-score optimality among intervals satisfying the same two tail requirements over the same admitted class.

\noindent\begin{minipage}{\linewidth}
\par\medskip\noindent\textit{Binary revisions.}
\begin{example}[Dependence robustness]\label{oa-ex:bernoulli}
\normalfont\par\noindent
Let $E,D\in\{0,1\}$ each have Bernoulli margin $p=0.12$ and $q=P(E=D=1)\in[0,p]$. Then $S=E+D$ satisfies
\[
 P(S=0)=1-2p+q,\quad P(S=1)=2p-2q,\quad P(S=2)=q.
\]
For $B=m+S$ and $m=-1.5$, independence ($q=p^2$) and historical direct calibration at $q=0$ give the short 90 percent equal-tail interval $[-1.5,-0.5]$, declaring deceleration. The dependence-robust interval $R$ gives $[-1.5,0.5]$ and leaves the direction open.
\end{example}
\end{minipage}\par

\begin{proof}[Proof of Example~\ref{oa-ex:bernoulli}]
At $q=p=0.12$, B has probabilities 0.88 at $m=-1.5$ and 0.12 at $m+2=0.5$. Expected score separates into convex endpoint objectives $-l+20\mathbb E(l-B)_+$ and $u+20\mathbb E(B-u)_+$. Their unique minima are the 0.05/0.95 quantiles m and $m+2$. Hence
\[
 \inf_J\sup_{q\in[0,p]}\mathbb E_q\IS(J;B)
 \ge\inf_J\mathbb E_{q=p}\IS(J;B)=2.
\]
R=$[m,m+2]$ scores 2 for every support value under every q, attaining the bound. The short interval $J_0=[m,m+1]$ has width one and misses only $m+2$, at distance one with probability q, so $\mathbb E_q\IS(J_0;B)=1+20q$.
\end{proof}

For intervals fixed before the errors, the dependence-robust prediction interval's expected-score difference is $1-20q$, so it has a lower expected score for $q>0.05$ and a higher expected score for $q<0.05$. The short interval breaches its upper-tail budget above 0.05, but its total coverage falls below 0.90 only above 0.10. Mean B is $-1.26$ and $P(B<0)=1-q\in[0.88,1]$, while the dependence-robust prediction interval leaves the direction open. Thus the wider interval attains the minimax value within this Bernoulli family.

\par\medskip\noindent\textit{Interval objectives.}
For uniform $[0,1]$ components, $m=0$ and $\alpha=0.10$, the minimax-score interval is $[0.1,1.9]$, whereas separate-tail requirements give $[0.05,1.95]$. The calculation uses convex ordering \citep{oa-bib-Kaas2002} and interval-score convexity \citep{oa-bib-Gneiting2007}, with fixed integrable component laws.

For component CDFs $F_E,F_D$, let $\Pi=\Pi(F_E,F_D)$ collect their joint laws and set $B^c=m+F_E^{-1}(U)+F_D^{-1}(U)$, where $U\sim\mathrm{Unif}(0,1)$. Convex ordering gives, for every finite fixed J,
\[
 \sup_{Q\in\Pi}\mathbb E_QIS_\alpha(J;B)=\mathbb E IS_\alpha(J;B^c).
\]
Pairing equal quantiles maximizes the score for every fixed interval. Endpoint minimizers satisfy $F_{B^c}(l^-)\le\tau\le F_{B^c}(l)$ and $F_{B^c}(u^-)\le1-\tau\le F_{B^c}(u)$ and can be nonunique.

For the uniform example, sharp lower-tail suprema are x on $[0,1]$, and upper-tail suprema are $2-x$ on $[1,2]$, giving $R^*=[0.05,1.95]$. The reflection is
\[
 D=\begin{cases}0.09-U,&U<0.09,\\1-U,&0.09\le U<0.91,\\1.91-U,&U\ge.91.\end{cases}
\]
Each band maps to itself with unit Jacobian magnitude. The sum has masses $0.09,0.82,0.09$ at $0.09,1,1.91$. Pairing equal quantiles gives the sum $2U$. For this sum, $[a,b]\subseteq[0,2]$ has expected score $(b-a)+20\{a^2/4+(2-b)^2/4\}$. Substitution yields $19/10$ and $77/40$, while the reflection gives 0.82 versus one coverage.

Strict-event suprema can have different extremizers or fail to be attained. Finite empirical fixed-event linear programs attain their extrema.

\subsection{Component accuracy}

\par\medskip\noindent\textit{Atoms and distribution proximity.}\label{oa-sec:atoms}\label{oa-sec:atom-illustration}\label{oa-sec:fragility-details}
Let $\delta(x)=\widehat F(x)-F(x)$, $h=\widehat F(x)-\widehat F(x^-)$ and $q=F(x)-F(x^-)$. The left limit is $F(x^-)=\lim_{z\uparrow x}F(z)$. For the E component suppose $d_K(\widehat F,F)\le r_E$, with $\rho=r_E+r_D$. Then
\[
\begin{aligned}
 h-q&=\delta(x)-\delta(x^-),\\
 |h-q|&\le|\delta(x)|+|\delta(x^-)|\le2d_K(\widehat F,F),\\
 q&\ge h-2r_E\ge h-2\rho,\qquad \rho=r_E+r_D.
\end{aligned}
\]
Write $h_E(x)=\widehat F_E(x)-\widehat F_E(x^-)$ and $h_D(x)=\widehat F_D(x)-\widehat F_D(x^-)$. If both current component distributions are continuous, necessarily
\[
 \rho\ge\tfrac12\{\max_x h_E(x)+\max_x h_D(x)\}.
\]

\par\medskip\noindent\textit{Historical completion.}\label{oa-sec:historical-sensitivity}
The following result extends the component-error analysis to a class that retains part of a historical dependence law. Let $C_0$ be a fixed joint law with uniform $(0,1)$ margins. For component CDFs $F,G$, $Q_{C_0}(F,G)$ is the law of $(F^{-1}(U),G^{-1}(V))$ when $(U,V)\sim C_0$. Take $F_E,F_D$ as the supplied baseline CDFs and $F_E^*,F_D^*$ as their current counterparts. Let $\lambda\in[0,1]$ give the share assigned to unrestricted dependence. The class is $(1-\lambda)Q_{C_0}(F_E,F_D)+\lambda\Pi(F_E,F_D)$. Let $C_0$ have rank-bin density $c_0\le\kappa_s$ generated by a matrix $M$ with marginal masses $p_i,q_j$. With component CDF discrepancies $r_E,r_D$, let $\rho=r_E+r_D$. If current anchored event probability differs from this transferred completion by at most $\zeta$, a decreasing event with baseline outer bounds $L_\lambda,U_\lambda$ has allowance
\begin{equation}\label{oa-eq:fragility}
 \varepsilon_\lambda=[\lambda+(1-\lambda)\kappa_s]\rho+(1-\lambda)\zeta.
\end{equation}
The resulting bounds are $[\max(0,L_\lambda-\varepsilon_\lambda),\min(1,U_\lambda+\varepsilon_\lambda)]$. Section~\ref{oa-sec:completion-proof} proves this fixed-event allowance. Bounded density is essential: for $E$ uniform on $(0,1)$ and $D=-E$, an arbitrarily small downward shift of $D$ changes $P(E+D<0)$ from zero to one, so a singular completion has no universal version of this modulus.

\phantomsection\label{oa-sec:completion-proof}

\begin{proposition}[Conditional fixed-event allowance]\label{oa-prop:completion}
Fix a measurable coordinatewise decreasing event H. Let M be a nonnegative probability matrix with $\sum_{ij}M_{ij}=1$, row sums $p_i>0$ and column sums $q_j>0$. On rank bins of those lengths, fix density $c_0=M_{ij}/(p_iq_j)$ bounded by $\kappa_s<\infty$ as the component quantile maps change. For $\lambda\in[0,1]$, let $L_\lambda,U_\lambda$ be baseline outer bounds for its replacement class. Current CDF discrepancies are at most $r_E,r_D\in[0,1]$. Suppose $\zeta\ge0$ and $Q^*=(1-\lambda)Q^*_{\rm anc}+\lambda Q^*_{\rm free}$, both components having the current margins, satisfy
\[
 |Q^*_{\rm anc}(H)-Q_{C_0}(F_E^*,F_D^*)(H)|\le\zeta.
\]
Then $Q^*(H)$ lies in the clipped bounds following Equation~(\ref{oa-eq:fragility}). At $\lambda=1$ no anchor is required and the allowance is $r_E+r_D$.
\end{proposition}

\begin{proof}[Proof of Proposition~\ref{oa-prop:completion}]
Let $I_i,J_j$ be rank bins of lengths $p_i,q_j$ and $c_0(u,v)=M_{ij}/(p_iq_j)$ on $I_i\times J_j$. Its margins are uniform since
\[
 \int c_0(u,v)dv=\sum_j M_{ij}/p_i=1,
\]
and analogously in u. For $A(F,G)=Q_{C_0}(F,G)(H)$, holding one coordinate fixed makes $H$ a threshold event in the other. Its CDF cutoff, using a left limit when the endpoint is excluded, changes by at most that coordinate's Kolmogorov discrepancy. Applying the density bound to each marginal in turn gives
\[
 |A(F_E^*,F_D)-A(F_E,F_D)|\le\kappa_s r_E,
 \qquad |A(F_E^*,F_D^*)-A(F_E,F_D)|\le\kappa_s\rho.
\]
Left limits handle atoms. No continuous-margin assumption is needed.

For unrestricted infima and suprema $T(F,G)$ over all couplings, the argument based on upper and lower CDF bounds in Appendix~\ref{main-app:proof-margin} of the main text gives $|T(F_E^*,F_D^*)-T(F_E,F_D)|\le\rho$. It applies to the present coordinatewise decreasing event with one-sided formal endpoint sections, without inferring total variation proximity for arbitrary perturbed laws from Kolmogorov proximity.

The anchor allowance is $\kappa_s\rho+\zeta$, the unrestricted allowance is $\rho$. Mixture weights give
\[
\begin{aligned}
 Q^*(H)&\in[L_\lambda-\varepsilon_\lambda,U_\lambda+\varepsilon_\lambda]\cap[0,1],\\
 \varepsilon_\lambda&=[(1-\lambda)\kappa_s+\lambda]\rho+(1-\lambda)\zeta.
\end{aligned}
\]
Valid baseline outer bounds suffice. Neither general attainment nor a common extremizer is needed.
\end{proof}

\subsection{Precision from additional revisions}

Here $p_E,p_D$ denote the target tail probabilities defined in Section~\ref{main-sec:rare-tail-learning} of the main text, and the aligned case gives the historical indicators those same means.

\phantomsection\label{oa-sec:rare-tail-extensions}

The main text isolates the information gain for one fixed tail bound. The calculations below show how historical misalignment and dependence across records alter that gain. For a data-selected threshold on a finite deterministic grid, a uniquely separated population minimizer and uniform consistency can yield selection of that minimizer with probability tending to one. First-order MSE conclusions additionally require a joint CLT and suitable moment or selection-error control.

If auxiliary revisions instead have mean $p_{D,A}=p_D+\Delta$, retaining the other sampling assumptions and target $p_E+p_D$, put $\omega=(N-n)/N$. The full estimate has bias $\omega\Delta$, revision variance $[nv_Y+(N-n)p_{D,A}(1-p_{D,A})]/N^2$ and overlap covariance $C_w/N$. Thus
\[
\begin{aligned}
 \operatorname{MSE}(\widehat U_P)-\operatorname{MSE}(\widehat U_F)
 ={}&(1/n-1/N)(v_Y+2C_w)\\
 &-\frac{N-n}{N^2}\{p_{D,A}(1-p_{D,A})-v_Y\}-\omega^2\Delta^2.
\end{aligned}
\]
The bias and variance changes can reverse the gain. The identity does not estimate misalignment or establish that the paired target is representative.

Without across-record independence, put $S_E=\sum_iw_iX_i$, $P=n^{-1}\sum_{i=1}^nY_i$ and $Z=(N-n)^{-1}\sum_{i=n+1}^NY_i$. Since $\widehat U_F=S_E+P-\omega(P-Z)$,
\[
 \Delta V=2\omega\operatorname{Cov}(S_E+P,P-Z)-\omega^2\operatorname{Var}(P-Z),
\]
whose sign is unrestricted. For X zero, fifty independent paired Bernoulli$(0.02)$ revisions and two hundred auxiliary copies of one independent Bernoulli$(0.02)$ variable have aligned means but
\[
 \operatorname{Var}(\widehat U_P)=v_Y/50,\qquad
 \operatorname{Var}(\widehat U_F)=(50+200^2)v_Y/250^2.
\]
The full/paired variance ratio is 32.04. Strong dependence among additional revisions can offset or reverse the precision gain from a larger history.

\FloatBarrier
\section{Monte Carlo experiments}\label{oa-sec:oa-simulations}

\subsection{Additional release records}\label{oa-sec:conditional-transfer-details}

The experiment separates the quantity of historical information from its relevance for current prediction. Holding the forecast, error history and evaluation outcomes fixed, it compares older records that follow the current revision law with the same number of records from a shifted historical source.

\par\medskip\noindent\textit{Data generation and timing.}
There are 128 independent paths. Each generates reference months $j=-199,\ldots,74$, with forecast date $s_j$ for month $j$ and calendar month $c_j=1+((j-1)\bmod12)$. The streams $Z_{m,j}$, $Z_{E,j}$ and $Z_{u,j}$ are independent standard normals. The simulated quantities satisfy
\[
m_{s_j}=Z_{m,j},\quad E_j=.5Z_{E,j},\quad u_j=.5Z_{u,j},\quad
A_j=m_{s_j}+E_j,
\]
\[
D_j=.25\sin\{2\pi(c_j-1)/12\}+.5A_j+u_j+
\delta\,1\{j\le0\},\qquad B_j=A_j+D_j.
\]
Both cases use the same simulated primitives and current outcomes. Stable transfer sets $\delta=0$. The historical-source case sets $\delta=1$ for the 200 older records, a shift of two innovation standard deviations that stress-tests transfer to the current target.

Evaluation uses twelve reference months $t=63,\ldots,74$. The forecast for month $t$ is made at $s=s_t$, before its first release. Error observations $E_1,\ldots,E_{t-1}$ supply the common equally weighted error law (62--73 observations). A revision enters only when $j+12<t$. Matched training therefore uses $j=1,\ldots,t-13$ (50--61 rows). Full training uses $j=-199,\ldots,t-13$ (250--261 rows). Every matched row is retained. Older forecasts and errors are not supplied. Earlier scored first-release errors can enter later origins, but no scored revision matures within this evaluation window.

\par\medskip\noindent\textit{Forecasts and references.}
Full and matched conditional forecasts fit the same intercept, first-release slope and eleven month indicators. Unweighted OLS and the Student predictive scale $\widehat\sigma\sqrt{1+x'(X'X)^+x}$, with residual degrees of freedom $n-13$, follow Section~\ref{oa-sec:bridge}. The revision regression control instead fits $D=\alpha+\gamma E+u$ on matched pairs. All three integrate over exactly the same empirical $E$ law. Thus full versus matched changes the revision-training records. Comparisons with the control additionally change the conditioning restriction.

The known-conditional-law reference is the equally weighted normal mixture with means $1.5(m_s+e_k)+\eta_{c_t}$ and scale $.5$. It knows the current conditional revision parameters but retains the feasible forecasts' estimated error law. The true-DGP reference also knows the error law:
\[
 B_t\mid m_s,c_t\sim N(1.5m_s+\eta_{c_t},13/16).
\]
The known-revision-law reference isolates uncertainty in the initial-error distribution, while the true-DGP reference also supplies that distribution. The source shift affects the full-history training records.

\par\medskip\noindent\textit{Numerical evaluation.}
Quantiles are evaluated from the analytic mixture at probabilities $.05,.50,.95$. The Monte Carlo intervals below describe variation across simulated paths.

The primary contrasts are full-minus-matched 90 percent interval scores under each revision law and their paired difference. Table~\ref{oa-tab:conditional-transfer-levels} reports scores, coverage, width and mean MSE. Table~\ref{oa-tab:conditional-transfer-allocation} separates the score contrast into midpoint and half-width contributions. Mean MSE evaluates the predictive mean against realized $B_t$.

Each path contributes its average over twelve origins. Means and paired contrasts use 128 independent path values. Their standard errors equal the sample standard deviation of those values divided by $\sqrt{128}$. Approximate pointwise 95 percent intervals use Student $t$ critical values with 127 degrees of freedom. The 1,536 origins and the paired cases share observations within paths.

\par\medskip\noindent\textit{Results.}
The stable-law gain comes mainly from midpoint placement, $-0.142$ versus $-0.004$ from half-width. The old-source shift gives a placement cost of 1.096, partly offset by a half-width contribution of $-0.084$. The allocations sum to full-minus-matched score. Stable full coverage is 0.886 versus 0.893 for matched: its score improves through placement and sharpness rather than higher coverage. Shifted full coverage falls to 0.827 even as its interval widens.

\begingroup\setlength{\tabcolsep}{4pt}
\begin{table}[htbp]\centering\small
\caption{Conditional predictions under stable and shifted histories}\label{oa-tab:conditional-transfer-levels}
\begin{tabular}{lrrrr}
\toprule Method/case & $\IS$ & Coverage & Width & Mean MSE\\\midrule
Full, stable & $3.862\;(0.087)$ & $0.886\;(0.0076)$ & $2.978\;(0.018)$ & $0.843\;(0.030)$\\
Full, historical shift & $5.020\;(0.11)$ & $0.827\;(0.0093)$ & $3.297\;(0.017)$ & $1.451\;(0.043)$\\
Matched pairs & $4.008\;(0.087)$ & $0.893\;(0.0078)$ & $3.113\;(0.022)$ & $0.900\;(0.031)$\\
Revision regression & $4.408\;(0.096)$ & $0.906\;(0.0071)$ & $3.522\;(0.031)$ & $1.107\;(0.039)$\\
Known revision law & $3.822\;(0.086)$ & $0.893\;(0.0075)$ & $2.950\;(0.017)$ & $0.830\;(0.030)$\\
Gaussian (True-DGP) & $3.791\;(0.087)$ & $0.895\;(0.0077)$ & $2.965$ & $0.820\;(0.030)$\\
\bottomrule\end{tabular}

\notes{Entries are means of path averages over twelve origins, with standard errors in parentheses. The known-law reference uses the empirical initial-error distribution. The true-DGP reference also knows that law and has constant analytical interval width.}
\end{table}
\begin{table}[htbp]\centering\small
\caption{Midpoint and half-width contributions to the score contrast}\label{oa-tab:conditional-transfer-allocation}
\begin{tabular}{llrrr}\toprule Case & Contribution & Mean & Standard error & 95\% interval\\\midrule
Stable law & Midpoint placement & $-0.142$ & $0.037$ & $[-0.215,-0.068]$\\
Stable law & Half-width & $-0.004$ & $0.016$ & $[-0.035,0.027]$\\
Stable law & Total score difference & $-0.146$ & $0.039$ & $[-0.223,-0.069]$\\
Historical-source shift & Midpoint placement & $1.096$ & $0.11$ & $[0.879,1.313]$\\
Historical-source shift & Half-width & $-0.084$ & $0.016$ & $[-0.115,-0.053]$\\
Historical-source shift & Total score difference & $1.012$ & $0.11$ & $[0.803,1.222]$\\
\bottomrule\end{tabular}
\notes{The order-averaged allocation replaces midpoint and half-width in both orders and averages their contributions. Intervals use the same 128 paired path means.}
\end{table}

\endgroup

\subsection{Publication delays}\label{oa-sec:release-simulation}

Earlier releases can recover information that delayed completed pairs omit. The simulation varies initial-error and revision distributions, their association and the release delay across stable, marginal-change, dependence-change, joint-change and initial-error-only designs. Each forecast faces recent $E$, delayed annual joint labels, older revision-only records and missing sources. The no-additional-delay arm releases $D$ sooner while holding generated outcomes and source masks fixed. Methods use only the records available under their publication calendar. Comparing calendars measures the effect of earlier releases for a given method.

For $k=1,2$, let $X_{kt}$ be a stationary Gaussian AR(1) process over reference months and $G_{ky}$ a stationary Gaussian AR(1) process over reference years. The four processes are mutually independent, with coefficient 0.5 and unit unconditional variance. The supplied forecast is zero, so $B_t=E_t+D_t$. We set
\[
 \begin{aligned}
 Z_{kt}&=\sqrt{0.75}X_{kt}+\sqrt{0.25}G_{k,\operatorname{year}(t)},\\
 E_t&=a_tZ_{1t},\quad D_t=b_t\{\rho_tZ_{1t}+\sqrt{1-\rho_t^2}Z_{2t}\}.
\end{aligned}
\]
Baselines for $(a,b,\rho)$ are $(1,0.5,-0.5)$ in the upward family. The downward family uses $(1,0.5,+0.5)$. Affected scales double or halve, dependence/joint cases reverse the correlation, and initial-error-only cases change a but not b. The simulated calendar runs from January 2018 to December 2021, with regime changes in January or July 2019. Pre/post phases split at the break and full retains the whole window. For each series, we simulate 500 independent paths and reuse the same underlying shocks across the distributional-change scenarios, upward and downward families, January and July 2019 break dates, and both release-delay settings. Stable refers to one-period unconditional margins and copula. Serial/reference-year shocks still affect conditional forecasts. 

\par\medskip\noindent\textit{Paired-path uncertainty.}\label{oa-sec:delay-details}
The two delay arms use the same generated paths, outcomes and source-availability masks. Within each series and design, score, coverage and width differences are averaged over the relevant origins within a path. Standard errors use the resulting 500 independent paired path values, preserving serial dependence and common shocks within each path. The simulated 2018--2021 evaluation window contains the January or July 2019 regime change. Pre- and post-change comparisons split at that break, while full-window comparisons retain all evaluation dates.

Table~\ref{oa-tab:delay-comparisons} reports the comparisons discussed in the main text. When only initial-error volatility rises, earlier revisions reduce completed-pair procedures' scores by 0.259--1.602 points. Full-history changes range from $-0.012$ to $+0.010$. Earlier pair formation restores access to recent initial errors that full histories already use. Under joint changes in component volatility and dependence, equal-weighted full-history robust intervals cover about 82--86 percent of outcomes at nominal 90 percent coverage.

Each delay arm gives calibrated width the development labels released by its cutoff, so its enlargement factor can change with timing. Each path scores 48 CPI months and 21 PCE months, with 27 PCE dates unavailable under the source mask. These simulation counts differ from the empirical 46/46/21/21 cohort. The fitted conditional procedure is studied in Section~\ref{oa-sec:conditional-transfer-details}.

\begin{table}[htbp]\centering\small\setlength{\tabcolsep}{4pt}
\caption{Release timing and coverage under changing component laws}\label{oa-tab:delay-comparisons}
\textit{Panel A. Completed-pair score changes}\par
\begin{tabular}{llrrrrr}
\toprule
Series & Break & \shortstack{Paired I\\O} & \shortstack{Paired R\\O} & \shortstack{Paired I\\W} & \shortstack{Paired R\\W} & Direct \\
\midrule
CPI & Jan. & \shortstack{$-0.880$\\[-1.5pt] $(0.0398)$} & \shortstack{$-0.372$\\[-1.5pt] $(0.0293)$} & \shortstack{$-1.602$\\[-1.5pt] $(0.0670)$} & \shortstack{$-0.562$\\[-1.5pt] $(0.0465)$} & \shortstack{$-1.258$\\[-1.5pt] $(0.0499)$} \\
CPI & Jul. & \shortstack{$-0.697$\\[-1.5pt] $(0.0340)$} & \shortstack{$-0.287$\\[-1.5pt] $(0.0246)$} & \shortstack{$-1.346$\\[-1.5pt] $(0.0604)$} & \shortstack{$-0.477$\\[-1.5pt] $(0.0415)$} & \shortstack{$-0.965$\\[-1.5pt] $(0.0412)$} \\
Core CPI & Jan. & \shortstack{$-0.850$\\[-1.5pt] $(0.0394)$} & \shortstack{$-0.314$\\[-1.5pt] $(0.0272)$} & \shortstack{$-1.475$\\[-1.5pt] $(0.0636)$} & \shortstack{$-0.486$\\[-1.5pt] $(0.0435)$} & \shortstack{$-1.216$\\[-1.5pt] $(0.0498)$} \\
Core CPI & Jul. & \shortstack{$-0.715$\\[-1.5pt] $(0.0345)$} & \shortstack{$-0.265$\\[-1.5pt] $(0.0242)$} & \shortstack{$-1.289$\\[-1.5pt] $(0.0575)$} & \shortstack{$-0.446$\\[-1.5pt] $(0.0404)$} & \shortstack{$-1.015$\\[-1.5pt] $(0.0422)$} \\
PCE & Jan. & \shortstack{$-0.939$\\[-1.5pt] $(0.0537)$} & \shortstack{$-0.432$\\[-1.5pt] $(0.0408)$} & \shortstack{$-1.538$\\[-1.5pt] $(0.0844)$} & \shortstack{$-0.574$\\[-1.5pt] $(0.0620)$} & \shortstack{$-1.336$\\[-1.5pt] $(0.0672)$} \\
PCE & Jul. & \shortstack{$-0.850$\\[-1.5pt] $(0.0493)$} & \shortstack{$-0.368$\\[-1.5pt] $(0.0388)$} & \shortstack{$-1.453$\\[-1.5pt] $(0.0789)$} & \shortstack{$-0.530$\\[-1.5pt] $(0.0579)$} & \shortstack{$-1.196$\\[-1.5pt] $(0.0619)$} \\
Core PCE & Jan. & \shortstack{$-0.815$\\[-1.5pt] $(0.0454)$} & \shortstack{$-0.324$\\[-1.5pt] $(0.0321)$} & \shortstack{$-1.282$\\[-1.5pt] $(0.0737)$} & \shortstack{$-0.462$\\[-1.5pt] $(0.0545)$} & \shortstack{$-1.236$\\[-1.5pt] $(0.0622)$} \\
Core PCE & Jul. & \shortstack{$-0.731$\\[-1.5pt] $(0.0418)$} & \shortstack{$-0.259$\\[-1.5pt] $(0.0290)$} & \shortstack{$-1.217$\\[-1.5pt] $(0.0701)$} & \shortstack{$-0.436$\\[-1.5pt] $(0.0525)$} & \shortstack{$-1.089$\\[-1.5pt] $(0.0587)$} \\
\bottomrule
\end{tabular}
\par\medskip\textit{Panel B. Full-history score changes}\par
\begin{tabular}{llrrrr}
\toprule
Series & Break & \shortstack{I\\O} & \shortstack{R\\O} & \shortstack{I\\W} & \shortstack{R\\W} \\
\midrule
CPI & Jan. & \shortstack{$0.001$\\[-1.5pt] $(0.0023)$} & \shortstack{$0.005$\\[-1.5pt] $(0.0027)$} & \shortstack{$-0.012$\\[-1.5pt] $(0.0017)$} & \shortstack{$0.005$\\[-1.5pt] $(0.0022)$} \\
CPI & Jul. & \shortstack{$0.002$\\[-1.5pt] $(0.0023)$} & \shortstack{$0.010$\\[-1.5pt] $(0.0026)$} & \shortstack{$-0.012$\\[-1.5pt] $(0.0018)$} & \shortstack{$0.004$\\[-1.5pt] $(0.0023)$} \\
Core CPI & Jan. & \shortstack{$0.003$\\[-1.5pt] $(0.0022)$} & \shortstack{$0.006$\\[-1.5pt] $(0.0030)$} & \shortstack{$-0.008$\\[-1.5pt] $(0.0016)$} & \shortstack{$0.006$\\[-1.5pt] $(0.0017)$} \\
Core CPI & Jul. & \shortstack{$0.005$\\[-1.5pt] $(0.0022)$} & \shortstack{$0.008$\\[-1.5pt] $(0.0028)$} & \shortstack{$-0.010$\\[-1.5pt] $(0.0017)$} & \shortstack{$0.007$\\[-1.5pt] $(0.0016)$} \\
PCE & Jan. & \shortstack{$0.001$\\[-1.5pt] $(0.0020)$} & \shortstack{$-0.001$\\[-1.5pt] $(0.0025)$} & \shortstack{$-0.007$\\[-1.5pt] $(0.0016)$} & \shortstack{$-0.002$\\[-1.5pt] $(0.0020)$} \\
PCE & Jul. & \shortstack{$0.002$\\[-1.5pt] $(0.0019)$} & \shortstack{$-0.006$\\[-1.5pt] $(0.0023)$} & \shortstack{$-0.006$\\[-1.5pt] $(0.0015)$} & \shortstack{$-0.002$\\[-1.5pt] $(0.0019)$} \\
Core PCE & Jan. & \shortstack{$0.000$\\[-1.5pt] $(0.0019)$} & \shortstack{$-0.002$\\[-1.5pt] $(0.0022)$} & \shortstack{$-0.008$\\[-1.5pt] $(0.0017)$} & \shortstack{$-0.003$\\[-1.5pt] $(0.0021)$} \\
Core PCE & Jul. & \shortstack{$0.001$\\[-1.5pt] $(0.0019)$} & \shortstack{$-0.005$\\[-1.5pt] $(0.0021)$} & \shortstack{$-0.008$\\[-1.5pt] $(0.0016)$} & \shortstack{$-0.001$\\[-1.5pt] $(0.0019)$} \\
\bottomrule
\end{tabular}
\par\medskip\textit{Panel C. Joint-change robust coverage}\par
\begin{tabular}{llrr}
\toprule
Series & Break & Actual delay & No additional delay \\
\midrule
CPI & Jan. & 0.833 (0.0025) & 0.847 (0.0023) \\
CPI & Jul. & 0.853 (0.0024) & 0.863 (0.0022) \\
Core CPI & Jan. & 0.836 (0.0026) & 0.849 (0.0024) \\
Core CPI & Jul. & 0.850 (0.0026) & 0.861 (0.0023) \\
PCE & Jan. & 0.820 (0.0034) & 0.830 (0.0032) \\
PCE & Jul. & 0.829 (0.0034) & 0.837 (0.0032) \\
Core PCE & Jan. & 0.825 (0.0033) & 0.835 (0.0031) \\
Core PCE & Jul. & 0.832 (0.0032) & 0.839 (0.0031) \\
\bottomrule
\end{tabular}
\notes{Panels A and B report no-additional-delay minus actual-delay 90 percent interval scores in the upward initial-error-volatility design. Panel C reports full-history equal-weight dependence-robust interval coverage in the upward joint scale-and-dependence design. Parentheses contain path-based standard errors from 500 paths. Score differences are paired across delay arms. O/W denote equal/recency initial-error weights, I independence and R the dependence-robust interval. Breaks occur in January or July 2019. The full phase contains 48 observed CPI months and 21 observed PCE months per path.}
\end{table}

\FloatBarrier
\section{Empirical design and forecasting results}\label{oa-sec:oa-application}

The full-versus-matched comparisons measure the contribution of additional release records within a conditional specification. Benchmark comparisons evaluate complete predictive approaches on the same later outcomes.

\subsection{Publications and information histories}

\par\smallskip\noindent\textbf{Forecast selection.}
For each reference month, we select the last finite Cleveland monthly forecast dated in that month before the chart's first actual point. The chart copy retrieved 18 September 2026 supplies 57 months per target, January 2022--September 2026. Forecast-date information is determined using publication dates.

\par\smallskip\noindent\textbf{Price-index vintages and release targets.}
We use CPI and PCE vintage levels with official BLS/BEA release dates. $A$ and $B$ each use one coherent vintage for $t,t-3,t-6$. The forecast uses one origin-admissible vintage for $t-1,t-3,t-6$. For CPI, $B$ is the following-year February update. For PCE, it is the following-year annual/comprehensive update.

\par\smallskip\noindent\textbf{Historical coverage and availability.}
First-release histories begin in June 2008 for CPI and July 2000 for PCE. Later-release and revision histories begin in January 2012 for CPI and July 2000 for PCE. Initial-error histories begin in August 2013 for CPI and September 2013 for PCE. At each origin, we retain all earlier reference periods whose relevant labels were published on a strictly prior Eastern calendar day. Histories expand as labels mature and can include release observations without a recorded own-origin forecast.

\par\smallskip\noindent\textbf{Width calibration.}
The development period is May 2014--December 2021. Width calibration uses annual labels published by the end-2021 cutoff.

\par\smallskip\noindent\textbf{Evaluation sample.}
In each PCE series, all 48 annual outcomes in 2022--2025 are mature. The 27 source-incomplete forecasts leave 21 scored dates. A forecast may exist on a release day while the required prior-day coherent level vintage is unavailable. Evaluation uses these 21 source-complete dates.

\par\smallskip\noindent\textbf{Forecast-date estimation.}
Specifications and weighting rules remain fixed across forecast dates. Every fit uses only observations published by its forecast date.

Of 228 candidate series--month observations, 154 have usable prediction inputs. These comprise 134 observed annual outcomes, with 46 per CPI series and 21 per PCE series, and 20 additional observations with known initial errors. The other 74 lack required inputs. The publication calendar distinguishes unavailable inputs from annual outcomes still awaiting release.

\begin{table}[htbp]
\centering\small
\caption{Available observations at forecast dates}\label{oa-tab:histories}
\begin{tabular}{lrrrrrr}
\toprule
Target & Months & Forecasts & Outcomes & Errors & Revisions & Pairs \\
\midrule
CPI & 57 & 52 & 46 & 123.5 & 132 & 113 \\
Core CPI & 57 & 52 & 46 & 123.5 & 132 & 113 \\
PCE & 57 & 25 & 21 & 53 & 258 & 49 \\
Core PCE & 57 & 25 & 21 & 53 & 258 & 49 \\
\bottomrule
\end{tabular}
\notes{Months counts all candidate reference months. Forecasts counts months with the inputs needed to construct a prediction, and Outcomes counts those with an observed later release. Errors, Revisions and Pairs are median historical counts at forecast dates. Counts do not measure independent effective sample size.}
\end{table}

\FloatBarrier
\subsection{Revision learning}\label{oa-sec:source-diagnostics}

The additional revision archive differs in source and period composition from forecast-linked records. In PCE, pre-2020 observations carry a median 86.05\% of full revision mass, compared with 67.35\% in paired histories. Records without their own forecast origin supply 64.73\% of full mass. Table~\ref{oa-tab:saved-fits} describes the fitted revision relationships, followed by diagnostics of source transfer and extrapolation.

Full-history fits predict PCE revisions more accurately in the conditional-revision diagnostic, which evaluates the forecast-date fit at the subsequently observed $A$. Mean absolute residuals are 0.4338 versus 0.5353 for headline PCE and 0.4056 versus 0.5145 for core PCE, with all three tail losses also favouring full histories. CPI differences are smaller and mixed. Median pinball loss equals half the absolute residual.

\begin{table}[htbp]
\centering\small\setlength{\tabcolsep}{4pt}
\caption{Revision relationships estimated at forecast dates}
\label{oa-tab:saved-fits}
\begin{tabular}{llrrrlr}
\toprule
Target & Fit & Origins & Training n & Median $\widehat\beta_A$ & Slope range & Median $\widehat\sigma$ \\
\midrule
CPI & Full & 46 & 108--156 & -0.1814 & $[-0.1847,-0.1757]$ & 0.5890 \\
CPI & Matched & 46 & 89--137 & -0.1678 & $[-0.1711,-0.1557]$ & 0.6094 \\
Core CPI & Full & 46 & 108--156 & -0.1371 & $[-0.1904,-0.1325]$ & 0.3163 \\
Core CPI & Matched & 46 & 89--137 & -0.1370 & $[-0.1927,-0.1324]$ & 0.3388 \\
PCE & Full & 21 & 234--282 & -0.1111 & $[-0.1134,-0.1082]$ & 0.6518 \\
PCE & Matched & 21 & 38--61 & -0.1525 & $[-0.1636,-0.1079]$ & 0.4686 \\
Core PCE & Full & 21 & 234--282 & -0.1396 & $[-0.1409,-0.1297]$ & 0.5031 \\
Core PCE & Matched & 21 & 38--61 & -0.0902 & $[-0.0952,-0.0685]$ & 0.3502 \\
\bottomrule
\end{tabular}
\notes{Both fits use the same 46/46/21/21 mature origins. Slopes are dimensionless. Residual scales are percentage points.}
\end{table}

All 268 full/matched fits have rank 13. Matched PCE/core-PCE samples have 25--48 residual degrees of freedom and maximum realized-input leverages of 1.200/1.085, compared with 221--269 and 0.055/0.066 for full histories. The extra records substantially strengthen the regression design. Revision regression provides the complementary parsimonious comparison, changing the conditioning restriction as well as the design.

We next examine the marginal forecasts and source-batch differences that affect transfer to the current target.

Component-distribution accuracy is mixed at the 0.05, 0.50 and 0.95 quantiles. Full initial-error histories lower mean pinball loss in 9/12 equal-weighted and 11/12 recency-weighted series--quantile comparisons. Full revision histories have six favourable and six unfavourable comparisons. These marginal diagnostics are distinct from the gains in conditional prediction. Among the four supplementary headline-PCE observations, recency-weighted full histories have higher pinball loss at the 0.05 and 0.50 quantiles but lower loss at 0.95.

Adding records changes slopes, month effects, residual scale and prediction leverage. The midpoint/width decomposition summarizes their combined predictive effect. Adding the issued forecast as a regressor does not reduce aggregate mean absolute residual in any of the eight target--weighting comparisons, although sparse tail cells are mixed. A release-batch-held-out diagnostic raises core-CPI mean absolute residual from 0.1762 before 2020 to 0.2914 after 2021. Headline-PCE residual means are 0.5525 in revision-only records versus 0.4097 in forecast-matched records. These source and period differences motivate the transfer analysis.

The implied first-release distribution also allows for values outside the historical range for the same calendar month. Across the 154 input-ready candidates with evaluable hypothetical-state laws---52 per CPI series and 25 per PCE series, including the 134 mature outcomes---mean original/full mass outside the same-month $A$-history envelope is 0.1524/0.1819/0.0746/0.1399, rising to 0.1561/0.1849/0.0800/0.1839 with recency weighting. Some configurations place almost all mass outside the observed envelope. Evaluating those predictions requires extrapolating the conditional relationship beyond the sample range.

\subsection{Comparison methods}\label{oa-sec:comparison-methods}\label{oa-sec:complete-benchmarks}\label{oa-sec:benchmark-methods}

Unless marked matched, linear and spline conditional predictions use full release histories. Independence, the dependence-robust interval, calibrated width and the equal-width diagnostic are denoted $I$, $R$, $K$ and $C$ in the calculations below. The historical-rank procedure provides directional probabilities and prediction intervals. Its probabilities use the tied-rank checkerboard masses $(\widehat\pi_{ij})$ defined in Section~\ref{oa-sec:operations}. Interval comparisons retain the issued historical-rank prediction intervals.

Direct calibration adds empirical quantiles of completed-pair total errors $S=E+D$ to $m_s$, transferring the historical total-error distribution to the current date.

For component masses $p_i,q_j$, couplings $\Pi(p,q)$ are nonnegative matrices with those row/column sums. Independence uses $p_iq_j$. Historical ranks transfer completed pairs' tied-rank association. Directional probabilities sum the tied-rank checkerboard masses over the specified event.

The historical-rank procedure weights pairs equally and spreads ties over their full rank intervals. The dependence-robust prediction interval limits each outside tail to 5 percent under every coupling, conditional on the supplied laws. Section~\ref{oa-sec:operations} gives the rank and endpoint constructions.

The revision regression fits $D=\beta_0+\beta_1 E+u$ and uses a Student predictive law conditional on the initial error. It uses only the $n_{P,s}$ released forecast-linked $(E,D)$ pairs. It then integrates the fitted conditional predictive law over \emph{all} $n_{E,s}$ eligible initial-error records, including $E$ records outside the paired sample. With $\delta_e$ denoting a unit mass at $e$, its mixing law at origin $s$ is
\[
 G_{E,s}=\frac{1}{n_{E,s}}\sum_{i=1}^{n_{E,s}}\delta_{e_{i,s}},\qquad
 F^{\mathrm{Stu}}_{B,s}(b)=\int F^{\mathrm{Stu}}_{D\mid E,s}(b-m_s-e\mid e)\,dG_{E,s}(e).
\]
Each record, including repeated values, has mass $1/n_{E,s}$. For a nondegenerate paired design, the Student component has $n_{P,s}-2$ degrees of freedom and predictive scale $\widehat\sigma\{1+1/n_{P,s}+(e-\bar E_{P,s})^2/S_{EE,P,s}\}^{1/2}$, where $\bar E_{P,s}$ is the paired-sample mean and $S_{EE,P,s}=\sum_{(e,d)\in\mathcal H^P_s}(e-\bar E_{P,s})^2$. Each completed-pair record contributes once, including records with identical values. The comparator uses the full $E$ membership with equal weights under both O and W, so its forecast repeats across those panels.

The revision regression requires at least three released pairs and a nonconstant $E$ predictor. A zero residual scale gives an atom at the fitted location with the record's original weight. The dependence-robust prediction interval uses the strict lower outside event and weak upper cumulative event for a closed interval. Its dependence robustness is conditional on adequate current component distributions. The equal-width diagnostic uses the independence interval's midpoint and the dependence-robust prediction interval's width to separate enlargement from placement.

Calibrated width multiplies the independence interval's half-width by a series/input factor, at least one, minimizing development mean interval score using labels released before the end-2021 cutoff. Six of eight factors equal one. Its post-2021 use follows the release calendar. Applying the same factor to earlier development dates is retrospective.

All benchmarks use the same 134 outcomes, with 46 each for CPI and core CPI and 21 each for PCE and core PCE. These observations span 46 related joint forecast origins. Predictions are evaluated against the specified later publications. The literature benchmarks generate their own forecasts, which are shared across the equal- and recency-weighted comparisons.

\paragraph{Four-series joint SV.}
The joint SV model is based on \citet{oa-bib-Clements2023}'s two-release architecture, rooted in \citet{oa-bib-Kishor2012}. We apply it to the four monthly inflation series, with the history, priors and initial-state treatment specified below. In the order CPI, core CPI, PCE, core PCE, let $\mathbf B_t$ contain later releases and $\mathbf r_t=\mathbf A_t-\mathbf B_t=-\mathbf D_t$. The model is
\begin{equation}\label{oa-eq:benchmark-joint}
\begin{aligned}
 \mathbf B_t&=c_B+\sum_{\ell=1}^{4}F_\ell\mathbf B_{t-\ell}+\epsilon_{B,t},\qquad
 \mathbf r_t=c_r+K\mathbf r_{t-1}+\epsilon_{r,t},\\
 L\epsilon_t&=\operatorname{diag}\{\exp(h_t/2)\}z_t,\quad z_t\sim N_8(0,I),\qquad
 h_t=h_{t-1}+\eta_t,\quad\eta_t\sim N_8(0,\Phi),
\end{aligned}
\end{equation}
where $\epsilon_t=(\epsilon_{B,t}',\epsilon_{r,t}')'$ and $L$ is unit lower triangular. Thus the innovation covariance is $L^{-1}\operatorname{diag}\{\exp(h_t)\}L^{-\prime}$. Both releases constrain $A=B+r$ and $B$ separately. A-only, B-only and empty months remain in the state calendar. At an origin for $t$, the history begins in June 2008 and ends at $t-1$. Current $A_t,B_t$ are unknown. Each release enters only if published on a strictly earlier Eastern calendar day, separately for each series. The model does not use Cleveland's nowcast.

The 88 coefficient means are zero. Lag variances are $\theta_b a_{ij}s_i^2/(\ell^2s_j^2)$ for block $b\in\{B,r\}$, with $a_{ii}=1$ and $a_{ij}=0.9$ otherwise. Intercept variances are $100s_i^2$. The $s_i^2$ are AR(1) residual scales estimated from the available history. The economic and revision shrinkage parameters, $\theta_B$ and $\theta_r$, are chosen separately by marginal likelihood over $\{0.01,\allowbreak 0.025,\allowbreak 0.05,\allowbreak 0.10,\allowbreak 0.15,\allowbreak 0.20,\allowbreak 0.25,\allowbreak 0.30,\allowbreak 0.35,\allowbreak 0.40,\allowbreak 0.45,\allowbreak 0.50,\allowbreak 0.75,\allowbreak 1,\allowbreak 2,\allowbreak 5\}$, taking the smallest parameter in an exact tie. The four initial eight-coordinate states have mean zero and covariance $100I_4\otimes\operatorname{diag}(s_1^2,\ldots,s_8^2)$. We use $h_0\sim N_8(0,100I)$, $\Phi\sim\operatorname{IW}_8(11,1.65I)$ and flat strictly lower loadings with full-rank conditionals. AR coefficients are unrestricted by stationarity.

At each forecast origin, we run one chain for 10,000 iterations, discard the first 5,000 and retain the remaining draws. Each iteration updates the joint coefficient block, triangular loadings and volatility, and latent releases. The volatility update uses a seven-component normal-mixture approximation to the log-squared Gaussian shocks. Diagnostics from the two halves of the retained chain, bulk and tail effective sample sizes and batch-quantile dispersion assess mixing and endpoint precision. All finite intervals from this fixed sampling schedule enter the comparison.

\paragraph{Real-time-vintage benchmark.}
The real-time-vintage benchmark follows \citet{oa-bib-Clements2017}'s direct revised-target projection and residual-Gaussian interval. It regresses $B_j$ on an intercept and two lagged momentum measurements, using reference months $t-144,\ldots,t-1$. Predictors use the latest raw vintage strictly before their own historical origin, and responses must mature before the current outer origin. The historical clock is the month's last Monday--Friday date, including holidays. All six levels for the two predictors must come from that vintage. Incomplete rows remain missing rather than compressing the calendar.

The lag pattern is $(1,2)$, except that CPI/core CPI in December 2025 use November and September, $(1,3)$, because October is unavailable. The same pattern applies to every training row for that case. Full-rank OLS gives $\widehat\sigma^2=\mathrm{RSS}/(n-3)$ and interval $\widehat B\pm z_{0.95}\widehat\sigma$, without parameter-uncertainty widening. Training sizes are 108--143 for CPI/core CPI and 54--62 for PCE/core PCE. The December $(1,3)$ cases use 133 rows. These are monthly revised-target intervals under the application's annual maturity calendar.

\subsection{Information gains}\label{oa-sec:geometry}

We separate changes in interval placement from changes in width for the linear conditional procedure.

Let $c_j=(l_j+u_j)/2$ and $h_j=(u_j-l_j)/2$ denote midpoint and half-width for full (F) and matched (M) predictions. Write $s(c,h;B)=2h+20(|B-c|-h)_+$. The original scores are $s_{FF}$ and $s_{MM}$. Crossed intervals give $s_{FM}=s(c_F,h_M;B)$ and $s_{MF}=s(c_M,h_F;B)$. Averaging the two replacement orders gives
\begin{equation}\label{oa-eq:geometry}
 \begin{split}
 \phi_c&=\tfrac12\{s_{FM}-s_{MM}+s_{FF}-s_{MF}\},\\
 \phi_h&=\tfrac12\{s_{MF}-s_{MM}+s_{FF}-s_{FM}\},\\
 \phi_c+\phi_h&=s_{FF}-s_{MM}.
 \end{split}
\end{equation}
The averaging splits the interaction equally. Negative contributions favour the full model. These diagnostic intervals combine features of two fitted forecasts. They are not additional forecasting procedures.

\begin{table}[htbp]\centering\small\setlength{\tabcolsep}{3pt}
\caption{Linear conditional interval-score decomposition}\label{oa-tab:crossed-scores}\label{oa-tab:contrast}\label{oa-tab:geometry}
\begin{tabular}{llrrrrrr}
\toprule
Target & Weights & $s_{MM}$ & $s_{FM}$ & $s_{MF}$ & $s_{FF}$ & Position & Width \\
\midrule
CPI & Equal & 4.3308 & 4.1609 & 4.3038 & 4.1340 & -0.1698 & -0.0270 \\
CPI & Recency & 4.2527 & 4.0632 & 4.2167 & 4.0272 & -0.1895 & -0.0360 \\
Core CPI & Equal & 3.1306 & 3.0799 & 3.1411 & 3.0796 & -0.0561 & 0.0051 \\
Core CPI & Recency & 3.2401 & 3.2071 & 3.2207 & 3.1877 & -0.0331 & -0.0194 \\
PCE & Equal & 4.0445 & 3.1310 & 3.9530 & 3.0141 & -0.9262 & -0.1042 \\
PCE & Recency & 4.1506 & 2.9548 & 4.1060 & 2.8949 & -1.2034 & -0.0523 \\
Core PCE & Equal & 5.0291 & 3.1140 & 4.3210 & 2.6668 & -1.7847 & -0.5777 \\
Core PCE & Recency & 4.6898 & 3.0262 & 4.1630 & 2.8807 & -1.4729 & -0.3361 \\
\midrule
\multicolumn{8}{l}{Coverage and date-level comparisons}\\
Target & Weights & $\Delta\IS$ & \multicolumn{2}{c}{$\Delta$ covered} & \multicolumn{2}{c}{Wins F/M} & \\
\midrule
CPI & Equal & -0.1968 & \multicolumn{2}{c}{0} & \multicolumn{2}{c}{45/1} & \\
CPI & Recency & -0.2255 & \multicolumn{2}{c}{0} & \multicolumn{2}{c}{46/0} & \\
Core CPI & Equal & -0.0510 & \multicolumn{2}{c}{-1} & \multicolumn{2}{c}{44/2} & \\
Core CPI & Recency & -0.0524 & \multicolumn{2}{c}{0} & \multicolumn{2}{c}{45/1} & \\
PCE & Equal & -1.0304 & \multicolumn{2}{c}{1} & \multicolumn{2}{c}{6/15} & \\
PCE & Recency & -1.2557 & \multicolumn{2}{c}{3} & \multicolumn{2}{c}{6/15} & \\
Core PCE & Equal & -2.3623 & \multicolumn{2}{c}{4} & \multicolumn{2}{c}{9/12} & \\
Core PCE & Recency & -1.8091 & \multicolumn{2}{c}{4} & \multicolumn{2}{c}{9/12} & \\
\bottomrule
\end{tabular}
\notes{Full minus matched for the linear conditional procedure on 46/46/21/21 series--month outcomes. The first crossed-score subscript selects midpoint and the second half-width. F and M denote full and matched histories. Position and width are score contributions, not physical changes. They sum to $\Delta\IS$ before rounding. Covered changes are counts. Wins count dates with lower full/matched score. Negative score differences favour full histories.}
\end{table}
\FloatBarrier

Full histories lower score in all eight constant-variance spline comparisons (Table~\ref{main-tab:practical} in the main text). Spline flexibility improves headline CPI and PCE but worsens the core-series full predictions. In core PCE the full/matched gap nevertheless grows because matched predictions deteriorate more. The annual comparisons show adverse headline-PCE groups in 2024 and 2025 under equal weighting for both specifications, and in 2024 under recency weighting. The spline comparison is also adverse under recency weighting in 2025.

Figure~\ref{oa-fig:score-chronology} shows when the full-history score gains accumulate across the scored dates.
\begin{figure}[htbp]\centering
\begin{minipage}{.48\linewidth}\centering\begin{tikzpicture}
\begin{axis}[width=\linewidth,height=4cm,title={CPI Inflation},xlabel={Months since January 2022},ylabel={Cumulative $\Delta IS$},grid=major]
\addplot[black,mark=*,mark size=.6pt] coordinates {(0,1.33231680894) (1,1.28512850457) (2,-1.02834689507) (3,-2.82071146274) (4,-2.86689197076) (5,-5.9534706664) (6,-6.00231314031) (7,-6.04784101667) (8,-6.09337711167) (9,-6.13344500823) (10,-6.17565972641) (11,-6.21859989434) (12,-6.27005094567) (13,-6.37307178404) (14,-6.47675127304) (15,-6.57949665763) (16,-6.68505055871) (17,-6.78890635008) (18,-6.89258800107) (19,-6.98778876433) (20,-7.08618914479) (21,-7.18539518162) (22,-7.28042347491) (23,-7.37826199463) (24,-7.48252965268) (25,-7.56710460198) (26,-7.65247726385) (27,-7.73696112312) (28,-7.82181945835) (29,-7.91049201787) (30,-7.99883443942) (31,-8.07648012922) (32,-8.15505353203) (33,-8.23479383019) (34,-8.31274401083) (35,-8.38964298972) (36,-8.47318166034) (37,-8.53901189078) (38,-8.60403833372) (39,-8.67030440219) (40,-8.73882114012) (41,-8.80465121911) (42,-8.86999990431) (43,-8.93063675119) (44,-8.99227252843) (47,-9.05348759513)};
\addplot[gray,dashed,mark=*,mark size=.6pt] coordinates {(0,-0.0833585220611) (1,-0.150711795451) (2,-2.47239901559) (3,-4.27864486354) (4,-4.34188895123) (5,-7.14674637798) (6,-7.20875757322) (7,-7.27328886936) (8,-7.33663972844) (9,-7.39533053937) (10,-7.4607918805) (11,-7.52941058757) (12,-7.60293781151) (13,-7.70349716924) (14,-7.80506361492) (15,-7.9068385968) (16,-8.01183608712) (17,-8.11473120301) (18,-8.2180179247) (19,-8.31396751472) (20,-8.41216394536) (21,-8.51047838715) (22,-8.60528949514) (23,-8.70297188019) (24,-8.80597569917) (25,-8.89019614981) (26,-8.97479301293) (27,-9.0584854622) (28,-9.14275934643) (29,-9.23067011168) (30,-9.31845705753) (31,-9.39598805621) (32,-9.47416288676) (33,-9.55363381203) (34,-9.63159206451) (35,-9.70864620014) (36,-9.79248461404) (37,-9.85848133225) (38,-9.92361512241) (39,-9.99005630519) (40,-10.0588834777) (41,-10.1248730281) (42,-10.1903214698) (43,-10.2506351269) (44,-10.3120071126) (47,-10.3729235255)};
\end{axis}\end{tikzpicture}\end{minipage}
\hfill
\begin{minipage}{.48\linewidth}\centering\begin{tikzpicture}
\begin{axis}[width=\linewidth,height=4cm,title={Core CPI Inflation},xlabel={Months since January 2022},ylabel={Cumulative $\Delta IS$},grid=major]
\addplot[black,mark=*,mark size=.6pt] coordinates {(0,-0.14532127555) (1,-0.181152826369) (2,-0.309191016665) (3,-0.742902720126) (4,-0.77664644022) (5,-0.810742035423) (6,-0.844766498684) (7,-0.875433606862) (8,-0.906339351833) (9,-0.936956237418) (10,-0.964828965964) (11,-0.990373440353) (12,-1.01870967296) (13,-1.08841153027) (14,-1.15841493141) (15,-1.2286297347) (16,-1.2989926976) (17,-1.36317763896) (18,-0.952445409792) (19,-1.01753421004) (20,-1.0814472766) (21,-1.14561728936) (22,-1.21034970479) (23,-1.27480423534) (24,-1.34354324116) (25,-1.40111965761) (26,-1.45887026864) (27,-1.51664565632) (28,-1.57448244182) (29,-1.51084129426) (30,-1.56970521489) (31,-1.62501232409) (32,-1.67982448873) (33,-1.73524228379) (34,-1.79059387542) (35,-1.8456324955) (36,-1.9040546038) (37,-1.95394449011) (38,-2.00391819678) (39,-2.05396321189) (40,-2.10390916886) (41,-2.15405990875) (42,-2.20410115842) (43,-2.25157285061) (44,-2.29930484289) (47,-2.34729383818)};
\addplot[gray,dashed,mark=*,mark size=.6pt] coordinates {(0,-0.0163675051446) (1,-0.0374251497714) (2,-0.288187350793) (3,-0.733171287627) (4,-0.753815205135) (5,-0.772431833199) (6,-0.789027880983) (7,-0.805892523746) (8,-0.820921143324) (9,-0.833681145292) (10,-0.84464921483) (11,-0.854296039463) (12,-0.863348823525) (13,-0.888083408093) (14,-0.913583764214) (15,-0.942814991036) (16,-0.97816489609) (17,-1.02132431145) (18,-1.0742482995) (19,-1.13204839672) (20,-1.19109473382) (21,-1.25056236248) (22,-1.31071258256) (23,-1.36992812276) (24,-1.43214582803) (25,-1.4859529) (26,-1.53964489087) (27,-1.59345037652) (28,-1.64665999089) (29,-1.59646340983) (30,-1.65361527949) (31,-1.70723486337) (32,-1.76050740328) (33,-1.81491822483) (34,-1.86944183062) (35,-1.92413500327) (36,-1.98218350268) (37,-2.02933131531) (38,-2.07642826728) (39,-2.12482567058) (40,-2.17309782132) (41,-2.2219641612) (42,-2.27089044544) (43,-2.31777852212) (44,-2.36511713201) (47,-2.41269999096)};
\end{axis}\end{tikzpicture}\end{minipage}
\par\medskip
\begin{minipage}{.48\linewidth}\centering\begin{tikzpicture}
\begin{axis}[width=\linewidth,height=4cm,title={PCE Inflation},xlabel={Months since January 2022},ylabel={Cumulative $\Delta IS$},grid=major]
\addplot[black,mark=*,mark size=.6pt] coordinates {(0,-2.15677569197) (1,-1.90955901019) (4,-11.2607717048) (7,-10.7552569523) (9,-10.409022465) (11,-9.85915979996) (12,-9.50433501827) (13,-25.0302752374) (16,-24.9011027422) (18,-24.7647518589) (21,-24.4674462574) (23,-24.0556086592) (24,-23.4023064659) (27,-23.4406224038) (30,-23.4895617676) (32,-23.1488308883) (34,-22.7636210738) (35,-22.3276931407) (38,-21.9291988499) (41,-22.0105665409) (44,-21.6385570212)};
\addplot[gray,dashed,mark=*,mark size=.6pt] coordinates {(0,-2.05431273557) (1,-1.81134134547) (4,-14.4566695119) (7,-13.9341797869) (9,-13.5739862751) (11,-13.0883679772) (12,-12.7607068134) (13,-28.095274365) (16,-27.9796144431) (18,-27.8587672633) (21,-27.5665426597) (23,-27.1763326475) (24,-26.8997106433) (27,-26.9167194135) (30,-26.9496444927) (32,-26.6140112052) (34,-26.2336414398) (35,-25.802287869) (38,-25.4062431116) (41,-26.743379249) (44,-26.3692776069)};
\end{axis}\end{tikzpicture}\end{minipage}
\hfill
\begin{minipage}{.48\linewidth}\centering\begin{tikzpicture}
\begin{axis}[width=\linewidth,height=4cm,title={Core PCE Inflation},xlabel={Months since January 2022},ylabel={Cumulative $\Delta IS$},grid=major]
\addplot[black,mark=*,mark size=.6pt] coordinates {(0,-5.7849120006) (1,-5.41077238673) (4,-26.4090532451) (7,-25.8996048418) (9,-25.5196368869) (11,-25.0588224621) (12,-28.0297388835) (13,-35.580519434) (16,-44.244879892) (18,-44.0115831497) (21,-43.8107272228) (23,-43.5245870606) (24,-46.0662871433) (27,-46.0902496452) (30,-46.1283935126) (32,-45.9038418278) (34,-45.654723036) (35,-45.3725210767) (38,-45.1162776825) (41,-49.8537502053) (44,-49.6085307914)};
\addplot[gray,dashed,mark=*,mark size=.6pt] coordinates {(0,-2.73522331431) (1,-2.37076713626) (4,-22.3756752626) (7,-21.8788732284) (9,-21.5140453992) (11,-21.087962348) (12,-23.0793941851) (13,-27.9251063885) (16,-33.2791901849) (18,-33.0996641532) (21,-32.9420835696) (23,-32.7026501897) (24,-34.5380142464) (27,-34.5750329302) (30,-34.6232230245) (32,-34.4134793088) (34,-34.171106219) (35,-33.8968442519) (38,-33.6440629945) (41,-38.2435440864) (44,-37.990421518)};
\end{axis}\end{tikzpicture}\end{minipage}
\par\medskip
\caption{Cumulative interval-score differences between full and matched histories}\label{oa-fig:score-chronology}
\notes{Cumulative full-minus-matched 90 percent interval scores for the linear conditional procedure. Solid black lines use equal weights. Dashed grey lines use recency weights. Negative values favour full histories. Points mark every scored date: 46/46/21/21 observations across the four series. Gaps in the PCE source history are retained.}
\end{figure}

\begin{table}[htbp]\centering\small
\caption{Full-minus-matched interval scores after omitting each year}\label{oa-tab:leave-one-year-out}
\begin{tabular}{llrrrr}
\toprule
Target & Weights & Omit 2022 & Omit 2023 & Omit 2024 & Omit 2025\\
\midrule
CPI & Equal & -0.083 & -0.232 & -0.237 & -0.233 \\
CPI & Recency & -0.084 & -0.271 & -0.276 & -0.270 \\
Core CPI & Equal & -0.040 & -0.061 & -0.052 & -0.051 \\
Core CPI & Recency & -0.046 & -0.056 & -0.055 & -0.053 \\
PCE & Equal & -0.785 & -0.496 & -1.558 & -1.240 \\
PCE & Recency & -0.885 & -0.819 & -1.850 & -1.433 \\
Core PCE & Equal & -1.637 & -2.076 & -3.184 & -2.521 \\
Core PCE & Recency & -1.127 & -1.758 & -2.453 & -1.883 \\
\bottomrule
\end{tabular}
\notes{Each entry is the full-minus-matched mean interval score for the linear conditional procedure after omitting the indicated reference year. Both fits use the same remaining outcomes.}
\end{table}
\FloatBarrier
\par\medskip\noindent\textit{Annual variation.}\label{oa-sec:annual-details}\label{oa-sec:supplementary-calculations}\label{oa-sec:annual-variation}
Monthly score differences share annual revision shocks. Index full-minus-comparator contrasts by $c$. Let $d_{ic}$ be the score difference on date i, $N_c$ the matched-date count, $n_{yc}$ the count in year y, and $\widehat\Delta_c=N_c^{-1}\sum_i d_{ic}$. The centered annual contribution is $u_{yc}=N_c^{-1}\sum_{i\in y}(d_{ic}-\widehat\Delta_c)$, and $u_y$ collects comparisons. With four years, define
\begin{equation}\label{oa-eq:annual-covariance}
 V_0=\frac43\sum_{y=1}^4u_yu_y'.
\end{equation}
The diagonal square roots of $V_0$ give descriptive annual variation in score units. Normalizing each contribution by the full count $N_c$ links this variation to the aggregate score contrast. The four annual groups retain dependence across related series and weighting rules.

\begin{table}[htbp]
\centering\small
\caption{Full-minus-matched scores and annual covariance scales}\label{oa-tab:score-covariance}
\begin{tabular}{llrrr}
\toprule
Target & Weights & $n$ & Mean & Annual scale \\
\midrule
CPI & Equal & 46 & -0.1968 & 0.1118 \\
CPI & Recency & 46 & -0.2255 & 0.1398 \\
Core CPI & Equal & 46 & -0.0510 & 0.0126 \\
Core CPI & Recency & 46 & -0.0524 & 0.0067 \\
PCE & Equal & 21 & -1.0304 & 0.6837 \\
PCE & Recency & 21 & -1.2557 & 0.7029 \\
Core PCE & Equal & 21 & -2.3623 & 0.9475 \\
Core PCE & Recency & 21 & -1.8091 & 0.7784 \\
\bottomrule
\end{tabular}
\notes{Mean is full minus matched for the linear conditional procedure. The annual scale is the square root of the corresponding diagonal entry of $V_0$.}
\end{table}

For core PCE against independence, the mean score differences are $-0.200$ under equal weighting and $-0.146$ under recency weighting, with annual scales 0.366 and 0.373. Revision regression has lower mean score than full-history conditional prediction in both core-CPI comparisons.

\FloatBarrier
\subsection{Benchmark performance}\label{oa-sec:benchmark-tables}

Table~\ref{oa-tab:benchmark-pooled} extends the selected comparisons in the main text to all ten procedures on the same outcomes. The linear and spline conditional procedures have the lowest pooled scores, while revision regression leads core CPI. The controlled full-versus-matched comparison separately measures the contribution of additional release records. Scores, widths and missed distances are percentage points. Severe means outside distance strictly above 0.5. A slash separates O/W. A single entry applies to both.

\clearpage
\begingroup\small\setlength{\tabcolsep}{4pt}
\begin{longtable}{lrrrr}
\caption{Predictive performance: pooled and by series}\label{oa-tab:benchmark-pooled}\\
\toprule
Method & $\IS$: O/W & Covered: O/W & Width: O/W & Severe: O/W\\\midrule
\endfirsthead
\toprule
Method & $\IS$: O/W & Covered: O/W & Width: O/W & Severe: O/W\\\midrule
\endhead
\multicolumn{5}{l}{Pooled ($n=134$)}\\*
Linear conditional & 3.367/3.382 & 118/124 & 2.323/2.648 & 6/3 \\*
Spline conditional & 3.365/3.352 & 118/124 & 2.260/2.576 & 6/3 \\*
Direct calibration & 3.524 & 113 & 2.155 & 6 \\*
Revision regression & 3.433 & 116 & 2.286 & 6 \\*
Independence & 3.472/3.674 & 125/126 & 2.710/3.093 & 5/4 \\*
Historical ranks & 3.380/3.441 & 118/122 & 2.293/2.671 & 6/4 \\*
Calibrated width & 3.510/3.718 & 125/126 & 2.774/3.154 & 5/4 \\*
Dependence-robust & 4.624/4.757 & 131/132 & 4.466/4.666 & 1/0 \\*
Real-time vintage & 4.845 & 117 & 3.657 & 5 \\*
Joint SV & 4.644 & 124 & 4.006 & 4 \\*
\midrule\noalign{\penalty0}
\multicolumn{5}{l}{CPI ($n=46$)}\label{oa-tab:benchmark-cpi}\\*
Linear conditional & 4.134/4.027 & 42/43 & 2.661/3.001 & 4/2 \\*
Spline conditional & 4.020/3.886 & 42/43 & 2.586/2.915 & 4/2 \\*
Direct calibration & 4.223 & 42 & 2.729 & 3 \\*
Revision regression & 4.292 & 42 & 2.933 & 3 \\*
Independence & 4.424/4.590 & 43 & 3.256/3.723 & 3/2 \\*
Historical ranks & 4.157/4.394 & 43/42 & 2.791/3.249 & 3/2 \\*
Dependence-robust & 5.693/5.803 & 44/45 & 5.415/5.678 & 1/0 \\*
Real-time vintage & 6.341 & 44 & 5.039 & 2 \\*
Joint SV & 6.217 & 43 & 5.278 & 2 \\*
\midrule\noalign{\penalty0}
\multicolumn{5}{l}{Core CPI ($n=46$)}\label{oa-tab:benchmark-corecpi}\\*
Linear conditional & 3.080/3.188 & 40/43 & 1.933/2.409 & 2/1 \\*
Spline conditional & 3.186/3.241 & 40/43 & 1.868/2.332 & 2/1 \\*
Direct calibration & 2.953 & 39 & 1.844 & 1 \\*
Revision regression & 2.779 & 40 & 1.869 & 1 \\*
Independence & 2.923/3.286 & 44/45 & 2.367/2.851 & 1 \\*
Historical ranks & 2.827/2.811 & 40/43 & 1.815/2.249 & 1 \\*
Calibrated width & 3.032/3.414 & 44/45 & 2.554/3.028 & 1 \\*
Dependence-robust & 4.154/4.464 & 45 & 3.971/4.324 & 0 \\*
Real-time vintage & 3.398 & 42 & 2.945 & 0 \\*
Joint SV & 3.442 & 42 & 2.808 & 2 \\*
\midrule\noalign{\penalty0}
\multicolumn{5}{l}{PCE ($n=21$)}\label{oa-tab:benchmark-pce}\\*
Linear conditional & 3.014/2.895 & 18/20 & 2.622/2.763 & 0 \\*
Spline conditional & 2.919/2.842 & 19/20 & 2.540/2.677 & 0 \\*
Direct calibration & 3.383 & 17 & 2.117 & 1 \\*
Revision regression & 3.257 & 18 & 2.301 & 1 \\*
Independence & 3.196/3.164 & 19 & 2.771/2.960 & 0 \\*
Historical ranks & 3.221/3.028 & 18/19 & 2.563/2.837 & 1/0 \\*
Dependence-robust & 4.345/4.285 & 21 & 4.345/4.285 & 0 \\*
Real-time vintage & 3.878 & 19 & 3.504 & 0 \\*
Joint SV & 4.413 & 20 & 4.384 & 0 \\*
\midrule\noalign{\penalty0}
\multicolumn{5}{l}{Core PCE ($n=21$)}\label{oa-tab:benchmark-corepce}\\*
Linear conditional & 2.667/2.881 & 18 & 2.141/2.286 & 0 \\*
Spline conditional & 2.768/2.937 & 17/18 & 2.126/2.267 & 0 \\*
Direct calibration & 3.384 & 15 & 1.618 & 1 \\*
Revision regression & 3.163 & 16 & 1.767 & 1 \\*
Independence & 2.867/3.027 & 19 & 2.203/2.377 & 1 \\*
Historical ranks & 3.048/3.145 & 17/18 & 1.982/2.162 & 1 \\*
Dependence-robust & 3.594/3.577 & 21 & 3.594/3.577 & 0 \\*
Real-time vintage & 5.708 & 12 & 2.341 & 3 \\*
Joint SV & 4.065 & 19 & 3.466 & 0 \\*
\bottomrule
\end{longtable}
\notes{Covered entries are counts. Denominators are 134 pooled, 46 per CPI series and 21 per PCE series. A single entry means O=W. Calibrated width equals independence date by date for CPI, PCE and core PCE under both weighting rules, so those duplicate rows are omitted. The pooled and core-CPI calibrated-width rows remain distinct.}
\endgroup

The 2022 outcomes distinguish sharpness from robustness. Under equal weighting, the linear and spline conditional procedures cover 27/36 and 26/36 outcomes with six severe misses each, compared with 33/36 and one for the dependence-robust prediction interval. Joint SV has lower mean outside distance than the linear conditional procedure, 0.032 versus 0.052/0.037 under O/W. Its O score disadvantage is 1.682 of additional width less 0.405 of lower miss penalty, approximately 1.278, using $\IS=\mathrm{width}+20\,\mathrm{outside\ distance}$. The dependence-robust prediction interval covers 131/132 pooled outcomes under O/W but has greater width and score than independence. Its equal-weight score is slightly below joint SV (4.624 versus 4.644). Under W it is higher (4.757).

\par\medskip\noindent\textit{Annual performance.}\label{oa-sec:benchmark-calendar}

Annual results show the score--coverage--width trade-off in each reference year. Cells report $(\IS,\ \mathrm{covered},\ \mathrm{width},\ \mathrm{severe})$, with denominators in the panel headings and score/width in percentage points. The equal-panel rows for Direct calibration, revision regression, the real-time-vintage benchmark and joint SV also apply under W. Recency panels report linear conditional, spline conditional, independence, historical ranks, calibrated width and dependence-robust prediction intervals.

\begingroup\small\setlength{\tabcolsep}{1pt}\thinmuskip=1mu
\begin{longtable}{lrrrr}
\caption{Annual score, coverage, width and severe misses}\label{oa-tab:benchmark-calendar}\\
\toprule Method & 2022 & 2023 & 2024 & 2025\\\midrule
\endfirsthead
\toprule Method & 2022 & 2023 & 2024 & 2025\\\midrule
\endhead
\multicolumn{5}{l}{Pooled, equal, $n=36,36,36,26$ respectively} \\*
Linear conditional & $(5.512,27,2.208,6)$ & $(2.606,33,2.410,0)$ & $(2.601,33,2.370,0)$ & $(2.509,25,2.299,0)$ \\*
Spline conditional & $(5.697,26,2.145,6)$ & $(2.512,33,2.338,0)$ & $(2.579,34,2.317,0)$ & $(2.406,25,2.235,0)$ \\*
Direct calibration & $(5.390,26,2.050,4)$ & $(3.111,29,2.145,1)$ & $(2.616,33,2.201,0)$ & $(2.770,25,2.252,1)$ \\*
Revision regression & $(5.164,27,2.166,4)$ & $(2.851,31,2.341,1)$ & $(2.726,33,2.331,0)$ & $(2.824,25,2.314,1)$ \\*
Independence & $(4.802,31,2.598,4)$ & $(2.977,35,2.786,0)$ & $(2.897,34,2.747,0)$ & $(3.113,25,2.706,1)$ \\*
Historical ranks & $(5.160,28,2.265,4)$ & $(2.809,32,2.318,1)$ & $(2.559,33,2.299,0)$ & $(2.843,25,2.293,1)$ \\*
Calibrated width & $(4.762,31,2.658,4)$ & $(3.042,35,2.851,0)$ & $(2.960,34,2.810,0)$ & $(3.184,25,2.777,1)$ \\*
Dependence-robust & $(4.760,33,4.171,1)$ & $(4.581,36,4.581,0)$ & $(4.569,36,4.569,0)$ & $(4.574,26,4.574,0)$ \\*
Real-time vintage & $(6.594,26,3.573,3)$ & $(3.716,35,3.696,0)$ & $(4.232,33,3.668,1)$ & $(4.837,23,3.703,1)$ \\*
Joint SV & $(6.296,33,4.951,2)$ & $(4.294,34,4.147,0)$ & $(3.847,34,3.513,1)$ & $(3.947,23,3.186,1)$ \\*
\midrule\noalign{\penalty0}
\multicolumn{5}{l}{Pooled, recency, $n=36,36,36,26$ respectively} \\*
Linear conditional & $(5.189,30,2.853,3)$ & $(2.988,35,2.911,0)$ & $(2.614,34,2.496,0)$ & $(2.489,25,2.212,0)$ \\*
Spline conditional & $(5.336,29,2.783,3)$ & $(2.820,36,2.820,0)$ & $(2.598,34,2.436,0)$ & $(2.387,25,2.146,0)$ \\*
Independence & $(5.005,32,3.341,3)$ & $(3.497,35,3.390,0)$ & $(2.955,34,2.903,0)$ & $(3.071,25,2.602,1)$ \\*
Historical ranks & $(5.159,29,2.954,3)$ & $(3.012,35,2.920,0)$ & $(2.607,34,2.497,0)$ & $(2.811,24,2.175,1)$ \\*
Calibrated width & $(5.016,32,3.416,3)$ & $(3.562,35,3.455,0)$ & $(3.005,34,2.954,0)$ & $(3.123,25,2.653,1)$ \\*
Dependence-robust & $(5.236,34,4.897,0)$ & $(4.914,36,4.914,0)$ & $(4.543,36,4.543,0)$ & $(4.172,26,4.172,0)$ \\*
\midrule\noalign{\penalty0}
\multicolumn{5}{l}{CPI, equal, $n=12,12,12,10$ respectively} \\*
Linear conditional & $(8.149,8,2.502,4)$ & $(2.786,12,2.786,0)$ & $(2.725,12,2.725,0)$ & $(2.623,10,2.623,0)$ \\*
Spline conditional & $(7.939,8,2.442,4)$ & $(2.682,12,2.682,0)$ & $(2.669,12,2.669,0)$ & $(2.545,10,2.545,0)$ \\*
Direct calibration & $(8.393,8,2.664,3)$ & $(2.752,12,2.752,0)$ & $(2.752,12,2.752,0)$ & $(2.751,10,2.751,0)$ \\*
Revision regression & $(8.001,8,2.790,3)$ & $(3.027,12,3.027,0)$ & $(3.011,12,3.011,0)$ & $(2.899,10,2.899,0)$ \\*
Independence & $(7.580,9,3.099,3)$ & $(3.371,12,3.371,0)$ & $(3.328,12,3.328,0)$ & $(3.218,10,3.218,0)$ \\*
Historical ranks & $(8.014,9,2.777,3)$ & $(2.831,12,2.831,0)$ & $(2.786,12,2.786,0)$ & $(2.766,10,2.766,0)$ \\*
Dependence-robust & $(6.031,10,4.965,1)$ & $(5.631,12,5.631,0)$ & $(5.588,12,5.588,0)$ & $(5.488,10,5.488,0)$ \\*
Real-time vintage & $(10.034,10,5.044,2)$ & $(5.168,12,5.168,0)$ & $(5.087,12,5.087,0)$ & $(4.820,10,4.820,0)$ \\*
Joint SV & $(10.525,10,6.929,2)$ & $(5.456,11,5.454,0)$ & $(4.627,12,4.627,0)$ & $(3.867,10,3.867,0)$ \\*
\midrule\noalign{\penalty0}
\multicolumn{5}{l}{CPI, recency, $n=12,12,12,10$ respectively} \\*
Linear conditional & $(7.071,9,3.135,2)$ & $(3.363,12,3.363,0)$ & $(2.910,12,2.910,0)$ & $(2.513,10,2.513,0)$ \\*
Spline conditional & $(6.798,9,3.078,2)$ & $(3.224,12,3.224,0)$ & $(2.847,12,2.847,0)$ & $(2.432,10,2.432,0)$ \\*
Independence & $(7.248,9,3.927,2)$ & $(4.164,12,4.164,0)$ & $(3.592,12,3.592,0)$ & $(3.108,10,3.108,0)$ \\*
Historical ranks & $(7.838,8,3.449,2)$ & $(3.662,12,3.662,0)$ & $(3.138,12,3.138,0)$ & $(2.649,10,2.649,0)$ \\*
Dependence-robust & $(6.401,11,5.919,0)$ & $(6.026,12,6.026,0)$ & $(5.639,12,5.639,0)$ & $(5.017,10,5.017,0)$ \\*
\midrule\noalign{\penalty0}
\multicolumn{5}{l}{Core CPI, equal, $n=12,12,12,10$ respectively} \\*
Linear conditional & $(5.494,9,1.783,2)$ & $(2.224,10,2.034,0)$ & $(2.479,11,1.984,0)$ & $(1.930,10,1.930,0)$ \\*
Spline conditional & $(5.750,9,1.702,2)$ & $(2.383,10,1.960,0)$ & $(2.508,11,1.930,0)$ & $(1.885,10,1.885,0)$ \\*
Direct calibration & $(5.149,8,1.695,1)$ & $(2.474,10,1.874,0)$ & $(2.103,11,1.906,0)$ & $(1.914,10,1.914,0)$ \\*
Revision regression & $(4.915,9,1.759,1)$ & $(2.072,10,1.943,0)$ & $(2.113,11,1.911,0)$ & $(1.862,10,1.862,0)$ \\*
Independence & $(4.410,10,2.279,1)$ & $(2.462,12,2.462,0)$ & $(2.396,12,2.396,0)$ & $(2.323,10,2.323,0)$ \\*
Historical ranks & $(4.982,9,1.750,1)$ & $(2.263,10,1.823,0)$ & $(2.073,11,1.866,0)$ & $(1.821,10,1.821,0)$ \\*
Calibrated width & $(4.290,10,2.460,1)$ & $(2.657,12,2.657,0)$ & $(2.585,12,2.585,0)$ & $(2.507,10,2.507,0)$ \\*
Dependence-robust & $(4.480,11,3.779,0)$ & $(4.075,12,4.075,0)$ & $(4.036,12,4.036,0)$ & $(3.998,10,3.998,0)$ \\*
Real-time vintage & $(4.331,9,2.833,0)$ & $(2.980,12,2.980,0)$ & $(3.198,11,2.961,0)$ & $(3.020,10,3.020,0)$ \\*
Joint SV & $(3.380,12,3.380,0)$ & $(3.482,11,3.041,0)$ & $(3.443,10,2.441,1)$ & $(3.469,9,2.283,1)$ \\*
\midrule\noalign{\penalty0}
\multicolumn{5}{l}{Core CPI, recency, $n=12,12,12,10$ respectively} \\*
Linear conditional & $(5.744,10,3.018,1)$ & $(2.661,12,2.661,0)$ & $(2.308,11,2.049,0)$ & $(1.807,10,1.807,0)$ \\*
Spline conditional & $(6.060,10,2.919,1)$ & $(2.568,12,2.568,0)$ & $(2.325,11,1.983,0)$ & $(1.763,10,1.763,0)$ \\*
Independence & $(5.256,11,3.585,1)$ & $(3.121,12,3.121,0)$ & $(2.431,12,2.431,0)$ & $(2.149,10,2.149,0)$ \\*
Historical ranks & $(5.196,10,3.045,1)$ & $(2.378,12,2.378,0)$ & $(1.872,11,1.870,0)$ & $(1.595,10,1.595,0)$ \\*
Calibrated width & $(5.288,11,3.808,1)$ & $(3.315,12,3.315,0)$ & $(2.582,12,2.582,0)$ & $(2.283,10,2.283,0)$ \\*
Dependence-robust & $(5.585,11,5.050,0)$ & $(4.695,12,4.695,0)$ & $(3.910,12,3.910,0)$ & $(3.507,10,3.507,0)$ \\*
\midrule\noalign{\penalty0}
\multicolumn{5}{l}{PCE, equal, $n=6,6,6,3$ respectively} \\*
Linear conditional & $(3.096,5,2.580,0)$ & $(3.454,5,2.659,0)$ & $(2.700,5,2.638,0)$ & $(2.600,3,2.600,0)$ \\*
Spline conditional & $(3.606,5,2.476,0)$ & $(2.805,5,2.607,0)$ & $(2.583,6,2.583,0)$ & $(2.445,3,2.445,0)$ \\*
Direct calibration & $(2.696,5,2.006,0)$ & $(5.090,4,2.064,1)$ & $(2.939,5,2.222,0)$ & $(2.234,3,2.234,0)$ \\*
Revision regression & $(2.550,5,2.244,0)$ & $(4.382,5,2.320,1)$ & $(3.297,5,2.320,0)$ & $(2.341,3,2.341,0)$ \\*
Independence & $(2.707,6,2.707,0)$ & $(3.962,5,2.814,0)$ & $(3.136,5,2.798,0)$ & $(2.761,3,2.761,0)$ \\*
Historical ranks & $(2.595,5,2.580,0)$ & $(4.364,5,2.553,1)$ & $(3.021,5,2.544,0)$ & $(2.587,3,2.587,0)$ \\*
Dependence-robust & $(4.061,6,4.061,0)$ & $(4.439,6,4.439,0)$ & $(4.502,6,4.502,0)$ & $(4.413,3,4.413,0)$ \\*
Real-time vintage & $(3.846,5,3.590,0)$ & $(3.574,6,3.574,0)$ & $(3.411,6,3.411,0)$ & $(5.481,2,3.374,0)$ \\*
Joint SV & $(5.206,6,5.206,0)$ & $(4.529,6,4.529,0)$ & $(3.763,6,3.763,0)$ & $(3.896,2,3.692,0)$ \\*
\midrule\noalign{\penalty0}
\multicolumn{5}{l}{PCE, recency, $n=6,6,6,3$ respectively} \\*
Linear conditional & $(2.611,6,2.611,0)$ & $(3.446,5,2.983,0)$ & $(2.774,6,2.774,0)$ & $(2.602,3,2.602,0)$ \\*
Spline conditional & $(3.072,5,2.495,0)$ & $(2.933,6,2.933,0)$ & $(2.719,6,2.719,0)$ & $(2.443,3,2.443,0)$ \\*
Independence & $(2.782,6,2.782,0)$ & $(3.852,5,3.208,0)$ & $(3.056,5,2.986,0)$ & $(2.772,3,2.772,0)$ \\*
Historical ranks & $(2.620,6,2.620,0)$ & $(3.673,5,3.122,0)$ & $(2.867,6,2.867,0)$ & $(2.874,2,2.643,0)$ \\*
Dependence-robust & $(3.951,6,3.951,0)$ & $(4.424,6,4.424,0)$ & $(4.502,6,4.502,0)$ & $(4.242,3,4.242,0)$ \\*
\midrule\noalign{\penalty0}
\multicolumn{5}{l}{Core PCE, equal, $n=6,6,6,3$ respectively} \\*
Linear conditional & $(2.688,5,2.098,0)$ & $(2.161,6,2.161,0)$ & $(2.499,5,2.161,0)$ & $(3.970,2,2.148,0)$ \\*
Spline conditional & $(3.198,4,2.105,0)$ & $(2.137,6,2.137,0)$ & $(2.536,5,2.121,0)$ & $(3.638,2,2.156,0)$ \\*
Direct calibration & $(2.563,5,1.575,0)$ & $(3.123,3,1.556,0)$ & $(3.046,5,1.665,0)$ & $(6.227,2,1.736,1)$ \\*
Revision regression & $(2.603,5,1.654,0)$ & $(2.525,4,1.787,0)$ & $(2.810,5,1.823,0)$ & $(6.264,2,1.844,1)$ \\*
Independence & $(2.126,6,2.126,0)$ & $(2.236,6,2.236,0)$ & $(2.798,5,2.238,0)$ & $(5.749,2,2.222,1)$ \\*
Historical ranks & $(2.374,5,1.953,0)$ & $(2.300,5,2.044,0)$ & $(2.615,5,1.942,0)$ & $(6.759,2,1.994,1)$ \\*
Dependence-robust & $(3.477,6,3.477,0)$ & $(3.634,6,3.634,0)$ & $(3.664,6,3.664,0)$ & $(3.605,3,3.605,0)$ \\*
Real-time vintage & $(6.988,2,2.093,1)$ & $(2.428,5,2.304,0)$ & $(5.411,4,2.501,1)$ & $(10.306,1,2.587,1)$ \\*
Joint SV & $(4.759,5,3.880,0)$ & $(3.361,6,3.361,0)$ & $(3.181,6,3.181,0)$ & $(5.856,2,3.416,0)$ \\*
\midrule\noalign{\penalty0}
\multicolumn{5}{l}{Core PCE, recency, $n=6,6,6,3$ respectively} \\*
Linear conditional & $(2.894,5,2.198,0)$ & $(2.431,6,2.431,0)$ & $(2.471,5,2.285,0)$ & $(4.573,2,2.171,0)$ \\*
Spline conditional & $(3.226,5,2.208,0)$ & $(2.402,6,2.402,0)$ & $(2.524,5,2.236,0)$ & $(4.259,2,2.177,0)$ \\*
Independence & $(2.242,6,2.242,0)$ & $(2.563,6,2.563,0)$ & $(2.629,5,2.389,0)$ & $(6.320,2,2.251,1)$ \\*
Historical ranks & $(2.265,5,2.115,0)$ & $(2.317,6,2.317,0)$ & $(2.755,5,2.102,0)$ & $(7.340,2,2.063,1)$ \\*
Dependence-robust & $(3.492,6,3.492,0)$ & $(3.620,6,3.620,0)$ & $(3.657,6,3.657,0)$ & $(3.500,3,3.500,0)$ \\*
\bottomrule
\end{longtable}
\notes{Each cell reports interval score, covered count, width and severe-miss count. Native benchmark forecasts are shared across the weighting rules. Calibrated width equals independence in CPI, PCE and core PCE. Related series, overlapping price levels and annual release batches make these descriptive calendar groups.}
\endgroup

\FloatBarrier
\section{Directional confidence and robust intervals}\label{oa-sec:oa-protection-evidence}

\subsection{Directional confidence}\label{oa-sec:june-probabilities}

\begin{table}[htbp]
\centering\small\setlength{\tabcolsep}{4pt}
\caption{June 2025 core PCE: probability that deceleration fails}\label{oa-tab:june-zero}
\begin{tabular}{lrrrrr}
\toprule
Input & Independence & \shortstack{Historical-rank\\probability} & \shortstack{Direct\\calibration} & Lower & Upper \\
\midrule
Equal & 0.016428 & 0.004073 & 0.000000 & 0.003704 & 0.085424 \\
Recency & 0.014181 & 0.003704 & 0.000000 & 0.003704 & 0.068251 \\
\bottomrule
\end{tabular}
\notes{The event is $B\ge0$, including zero. Independence, historical-rank probability and the bounds share the full component margins. Direct calibration uses paired total errors. Compare the upper bound with 0.05 and 0.10 for the one-sided requirements.}
\end{table}

\begin{table}[htbp]
\centering\small\setlength{\tabcolsep}{4pt}
\caption{Tail probabilities underlying June's upper bound}\label{oa-tab:june-zero-support}
\begin{tabular}{lrrrrrr}
\toprule
Input & $e_k$ & $P(E>e_k)$ & $D$ full & $D$ paired & Sum & Aux. \\
\midrule
Equal & 0.816062 & 0.048387 & 0.037037 & 0.018182 & 0.085424 & 9 \\
Recency & 0.816062 & 0.031214 & 0.037037 & 0.018182 & 0.068251 & 9 \\
\bottomrule
\end{tabular}
\notes{The revision event is $g(e_k,D)\ge0$. Aux. counts revision-only memberships. Full and paired D evaluate this same predicate. The interval-endpoint witnesses use different predicates.}
\end{table}

June's sharp upper non-deceleration probability lies between 0.05 and 0.10 under both rules. All-dependence deceleration therefore meets 90 percent but not 95 percent, and the later acceleration contradicts the retained 90 percent assertion. The dependence-robust prediction interval contains zero and addresses a separate two-tail requirement. Direct calibration's different total-error law lies outside the shared-margin sharp range in twelve input-ready cases. Individual contradictions provide outcome checks rather than calibration tests.

With recency weighting, the paired revision distribution gives a tail sum of about 0.0494 at this same threshold, sufficient for a 95 percent deceleration assertion under every coupling of the unchanged initial-error margin and the paired revision margin. The full revision distribution gives a sharp upper non-deceleration probability of 0.068251, leading the dependence bound to withhold that assertion. Here, using the additional revision records withdraws a 95 percent assertion contradicted by the later acceleration.\footnote{The equal-weighted tail sum at this threshold is about 0.0666. The displayed calculation therefore establishes the paired-history 95 percent assertion under recency weighting.}

\subsection{Component accuracy}\label{oa-sec:margin-sensitivity}

Table~\ref{oa-tab:margin-sensitivity} applies the sufficient bounds to deceleration assertions when the sum of the two component CDF errors is assumed to be at most one probability percentage point. It reports assertions supported under every compatible dependence pattern and independence--dependence contrasts that persist throughout this component-error class. Section~\ref{oa-sec:atoms} explains the implications of discrete component distributions.

\begin{table}[htbp]
\centering\small\setlength{\tabcolsep}{4pt}
\caption{Assertions allowing one-point CDF errors}\label{oa-tab:margin-sensitivity}
\begin{tabular}{llrrrrr}
\toprule
 & & & \multicolumn{2}{c}{95 percent requirement} & \multicolumn{2}{c}{90 percent requirement} \\
\cmidrule(lr){4-5}\cmidrule(lr){6-7}
Target & Weights & $n$ & \shortstack{Dependence-\\sensitive} & Robust & \shortstack{Dependence-\\sensitive} & Robust \\
\midrule
CPI & Equal & 46 & $4\to3$ & $8\to4$ & $6\to5$ & $8\to8$ \\
CPI & Recency & 46 & $4\to2$ & $7\to7$ & $4\to3$ & $9\to9$ \\
Core CPI & Equal & 46 & $6\to3$ & $5\to2$ & $5\to5$ & $9\to7$ \\
Core CPI & Recency & 46 & $4\to3$ & $3\to3$ & $2\to1$ & $9\to9$ \\
PCE & Equal & 21 & $3\to2$ & $3\to3$ & $4\to4$ & $4\to4$ \\
PCE & Recency & 21 & $2\to1$ & $3\to3$ & $4\to3$ & $3\to3$ \\
Core PCE & Equal & 21 & $3\to2$ & $1\to1$ & $2\to2$ & $2\to2$ \\
Core PCE & Recency & 21 & $2\to2$ & $1\to1$ & $2\to2$ & $2\to2$ \\
\bottomrule
\end{tabular}
\notes{Arrows compare the original count with the count for which the conclusion follows uniformly when the sum of the two component CDF errors is at most 0.01. Dependence-sensitive counts contrasts between independence support and dependence-bound withdrawal. Robust counts support under every compatible dependence pattern. The counts use the sufficient bounds and retain the strict/weak event conventions.}
\end{table}

\par\medskip\noindent\textit{June component accuracy.}
With June's counts $n_E=62$ and $n_D=270$, equal record weights imply the necessary continuous-law floor $1/(2n_E)+1/(2n_D)=83/8370$, about 0.992 probability percentage points. Repeated equal values can increase it. Continuity is a benchmark, not a requirement for rounded official data. For June 2025 core PCE, the recency-weighted initial-error distribution has an atom of 11.65048\% near $-0.127$ percentage points. Combined distance at most one probability percentage point requires at least 9.65048\% current mass at that \emph{same exact value}. A continuous current $E$ law alone requires at least 5.82524 points of distance.

The one-percentage-point exercise states which assertions follow if its CDF-error bound holds. The atom calculation shows what that bound requires of the current component distributions.\par\medskip\noindent\textit{Endpoint-tail diagnostics.}
At each issued robust endpoint, a minimizing threshold splits the outside event into an initial-error tail and a revision tail. Their supplied probability sums satisfy the endpoint budgets at every evaluated origin. Establishing endpoint extremality additionally requires checking that no closer cut satisfies those budgets. In headline PCE, the supplied zero initial-error tail is exceeded on two of 21 dates under each weighting rule. Component exceedances need not coincide with an interval miss.

At these same full-history-selected predicates, squared losses compare the full and paired revision-tail probabilities and the predicted sum of initial-error and revision exceedance indicators. This count takes values zero, one or two. It is distinct from the union of the component events and from an outside event for $B$. Headline-PCE revision-tail loss worsens under both weighting rules even though the combined count loss improves. Six of the eight series--weighting count comparisons favour full revisions, compared with four revision-only comparisons. Tail-count accuracy therefore concerns a different target from current component-distribution accuracy.

\subsection{Interval performance}\label{oa-sec:benchmark-june}

Table~\ref{oa-tab:original-methods} reports coverage, width and directional conclusions. Miss denotes mean distance outside the interval, and C/F/A counts correct/incorrect/withheld declarations. Table~\ref{main-tab:directional} in the main text reports one-sided deceleration assertions.

\begingroup\small\setlength{\tabcolsep}{4pt}
\begin{longtable}{llrrrr}
\caption{History comparisons and interval performance by series}\label{oa-tab:original-methods}\\
\toprule
Input & Construction & Covered & Miss & Width & C/F/A \\
\midrule
\endfirsthead
\toprule
Input & Construction & Covered & Miss & Width & C/F/A \\
\midrule
\endhead
\multicolumn{6}{r}{Continued on next page}\\
\endfoot
\bottomrule
\endlastfoot
\multicolumn{6}{l}{CPI ($n=46$)}\label{oa-tab:original-method-cpi}\\*
Equal & Independence & 43/46 & 0.058 & 3.256 & 20/3/23 \\*
Equal & Dependence-robust & 44/46 & 0.014 & 5.415 & 8/0/38 \\*
Equal & Linear conditional & 42/46 & 0.074 & 2.661 & 20/2/24 \\*
Equal & Linear conditional, matched & 42/46 & 0.080 & 2.735 & 21/3/22 \\*
Recency & Independence & 43/46 & 0.043 & 3.723 & 18/2/26 \\*
Recency & Dependence-robust & 45/46 & 0.006 & 5.678 & 7/0/39 \\*
Recency & Linear conditional & 43/46 & 0.051 & 3.001 & 20/2/24 \\*
Recency & Linear conditional, matched & 43/46 & 0.059 & 3.080 & 20/2/24 \\*
\midrule\noalign{\penalty0}
\multicolumn{6}{l}{Core CPI ($n=46$)}\label{oa-tab:original-method-corecpi}\\*
Equal & Independence & 44/46 & 0.028 & 2.367 & 19/0/27 \\*
Equal & Dependence-robust & 45/46 & 0.009 & 3.971 & 6/0/40 \\*
Equal & Linear conditional & 40/46 & 0.057 & 1.933 & 20/0/26 \\*
Equal & Linear conditional, matched & 41/46 & 0.057 & 1.984 & 20/0/26 \\*
Recency & Independence & 45/46 & 0.022 & 2.851 & 14/0/32 \\*
Recency & Dependence-robust & 45/46 & 0.007 & 4.324 & 4/0/42 \\*
Recency & Linear conditional & 43/46 & 0.039 & 2.409 & 14/0/32 \\*
Recency & Linear conditional, matched & 43/46 & 0.040 & 2.449 & 14/0/32 \\*
\midrule\noalign{\penalty0}
\multicolumn{6}{l}{PCE ($n=21$)}\label{oa-tab:original-method-pce}\\*
Equal & Independence & 19/21 & 0.021 & 2.771 & 8/1/12 \\*
Equal & Dependence-robust & 21/21 & 0.000 & 4.345 & 3/0/18 \\*
Equal & Linear conditional & 18/21 & 0.020 & 2.622 & 8/1/12 \\*
Equal & Linear conditional, matched & 17/21 & 0.086 & 2.326 & 9/2/10 \\*
Recency & Independence & 19/21 & 0.010 & 2.960 & 7/1/13 \\*
Recency & Dependence-robust & 21/21 & 0.000 & 4.285 & 3/0/18 \\*
Recency & Linear conditional & 20/21 & 0.007 & 2.763 & 5/0/16 \\*
Recency & Linear conditional, matched & 17/21 & 0.084 & 2.472 & 10/2/9 \\*
\midrule\noalign{\penalty0}
\multicolumn{6}{l}{Core PCE ($n=21$)}\label{oa-tab:original-method-corepce}\\*
Equal & Independence & 19/21 & 0.033 & 2.203 & 8/1/12 \\*
Equal & Dependence-robust & 21/21 & 0.000 & 3.594 & 1/0/20 \\*
Equal & Linear conditional & 18/21 & 0.026 & 2.141 & 7/1/13 \\*
Equal & Linear conditional, matched & 14/21 & 0.158 & 1.868 & 9/1/11 \\*
Recency & Independence & 19/21 & 0.032 & 2.377 & 6/1/14 \\*
Recency & Dependence-robust & 21/21 & 0.000 & 3.577 & 1/0/20 \\*
Recency & Linear conditional & 18/21 & 0.030 & 2.286 & 5/1/15 \\*
Recency & Linear conditional, matched & 14/21 & 0.133 & 2.036 & 9/1/11 \\*
\end{longtable}

\endgroup

Table~\ref{oa-tab:benchmark-june} shows how the procedures handle June 2025 core PCE's later acceleration, $B=0.117$ percentage point. The dependence-robust interval covers the outcome and includes zero under both weighting rules. Every other listed interval lies below zero and declares deceleration. Sections~\ref{oa-sec:june-probabilities} and \ref{oa-sec:margin-sensitivity} give the corresponding one-sided confidence calculations.

\begin{table}[htbp]\centering\small\setlength{\tabcolsep}{3pt}
\caption{June 2025 core PCE intervals}\label{oa-tab:benchmark-june}\label{oa-tab:june2025}
\begin{tabular}{lrrrrr}
\toprule
 & \multicolumn{2}{c}{Equal} & \multicolumn{2}{c}{Recency} & \\
\cmidrule(lr){2-3}\cmidrule(lr){4-5}
Method & Interval & $\IS$ & Interval & $\IS$ & Severe \\
\midrule
Linear conditional & $[-2.306,-0.156]$ & 7.616 & $[-2.399,-0.243]$ & 9.362 & 0 \\
Linear conditional, matched & $[-2.323,-0.405]$ & 12.354 & $[-2.410,-0.485]$ & 13.961 & 1 \\
Spline conditional & $[-2.283,-0.105]$ & 6.622 & $[-2.381,-0.195]$ & 8.431 & 0 \\
Direct calibration & $[-2.294,-0.557]$ & 15.212 & $[-2.294,-0.557]$ & 15.212 & 1 \\
Revision regression & $[-2.376,-0.546]$ & 15.091 & $[-2.376,-0.546]$ & 15.091 & 1 \\
Independence & $[-2.637,-0.412]$ & 12.806 & $[-2.728,-0.493]$ & 14.442 & 1 \\
Historical ranks & $[-2.533,-0.598]$ & 16.230 & $[-2.678,-0.674]$ & 17.832 & 1 \\
Dependence-robust & $[-3.266,0.343]$ & 3.609 & $[-3.241,0.246]$ & 3.487 & 0 \\
Real-time vintage & $[-3.528,-0.966]$ & 24.215 & $[-3.528,-0.966]$ & 24.215 & 1 \\
Joint SV & $[-4.215,-0.249]$ & 11.285 & $[-4.215,-0.249]$ & 11.285 & 0 \\
\bottomrule
\end{tabular}
\notes{The later outcome is approximately $+0.117$ percentage point. Only the dependence-robust interval covers it under either rule. Severe denotes a miss exceeding 0.5 percentage point. Calibrated width equals independence for this series under both rules. Native benchmark forecasts repeat across rules.}
\end{table}
\par\medskip\noindent\textit{Annual comparisons.}
The dependence-robust prediction interval has the lowest score among the eight reference procedures in headline CPI in 2022 and over the three core-PCE observations in 2025 taken together, under both rules. Fewer severe misses offset width in these four target/input/year groups. Independence beats the dependence-robust prediction interval in the other 28, including core CPI in 2022. Table~\ref{oa-tab:benchmark-calendar} reports the annual method scores, using 12/12/12/10 observations for each CPI series and 6/6/6/3 for each PCE series in 2022--2025.

\par\medskip\noindent\textit{Equal-width comparison.}
Wider intervals account for much of the dependence-robust prediction interval’s coverage gains. In development, the independence interval’s median worst-case coverage shortfall below 90\% is 28.72/25.11 probability percentage points under O/W. The equal-width diagnostic, which borrows the width of the dependence-robust prediction interval but keeps the independence midpoint, reduces the shortfall to 0.35/0.22. The remaining difference isolates placement.

On the evaluation sample, the equal-width diagnostic and dependence-robust interval have identical coverage and score in both PCE series. Their CPI differences are small and vary in sign (Table~\ref{oa-tab:equal-width}). The additional PCE coverage is therefore reproduced by the robust interval's width at the independence midpoint.

\begin{table}[htbp]\centering\small\setlength{\tabcolsep}{4pt}
\caption{Equal-width interval diagnostic}\label{oa-tab:equal-width}
\begin{tabular}{llrrrrrr}
\toprule
 & & & & \multicolumn{2}{c}{Dependence-robust} & \multicolumn{2}{c}{\shortstack{Equal-width\\diagnostic}}\\
\cmidrule(lr){5-6}\cmidrule(lr){7-8}
Target & Weights & $N$ & Width & Covered & Score & Covered & Score\\
\midrule
CPI & Equal & 46 & 5.415 & 44 & 5.693 & 45 & 5.563 \\
CPI & Recency & 46 & 5.678 & 45 & 5.803 & 45 & 5.738 \\
Core CPI & Equal & 46 & 3.971 & 45 & 4.154 & 45 & 4.163 \\
Core CPI & Recency & 46 & 4.324 & 45 & 4.464 & 45 & 4.457 \\
PCE & Equal & 21 & 4.345 & 21 & 4.345 & 21 & 4.345 \\
PCE & Recency & 21 & 4.285 & 21 & 4.285 & 21 & 4.285 \\
Core PCE & Equal & 21 & 3.594 & 21 & 3.594 & 21 & 3.594 \\
Core PCE & Recency & 21 & 3.577 & 21 & 3.577 & 21 & 3.577 \\
\bottomrule
\end{tabular}
\notes{At each date, the diagnostic combines the independence midpoint with the dependence-robust interval's width. All entries use the same outcomes within each series and weighting rule.}
\end{table}

\par\medskip\noindent\textit{Component-history comparison.}
Changing both initial-error and revision histories answers a different question from the full-versus-matched conditional fits, which share their initial-error input. Under equal weighting, full component histories add one covered headline-PCE outcome and two covered core-PCE outcomes under independence, while lowering mean score by 0.1385 and 0.2238 points. For the core-PCE robust interval, paired and full inputs both cover all 21 outcomes. Full inputs increase width and score by 0.7057 under equal weighting and 0.5725 under recency weighting. Equal-weighted correct directional declarations fall from five to one. Additional records change the distributions to which the robustness requirement is applied.

Under the reconstructed historical-rank law, independence has lower expected interval score than the dependence-robust interval in every overall and annual comparison. This is an expectation under the reconstructed law.

\FloatBarrier
\section{Numerical conventions}\label{oa-sec:numerical}

\par\medskip\noindent\textit{Historical ranks.}\label{oa-sec:operations}

For an empirical CDF $\widehat F$, $\widehat F(x^-)$ is its left limit, an atom has mass $\widehat F(x)-\widehat F(x^-)$, and the generalized inverse is $\widehat F^{-1}(u)=\inf\{x:\widehat F(x)\ge u\}$. Ties retain their observed masses without interpolation or randomization.

Let $n=|\mathcal H^P_s|$, and let $F_{E,P},F_{D,P}$ be the empirical CDFs among these released forecast-linked pairs. For pair $k$, form the intervals occupied by its tied $E$ and $D$ values in those paired CDFs:
\[
 U_k=[F_{E,P}(e_k^-),F_{E,P}(e_k)],\qquad
 V_k=[F_{D,P}(d_k^-),F_{D,P}(d_k)].
\]
Let $e_i,d_j$ be the full component support values with masses $p_i,q_j$, and $I_i,J_j$ their rank bins. For an interval $I$, $|I|$ denotes its length. The historical tied-rank coupling allocates each pair's uniform rectangle to these full-law bins by
\[
 \widehat\pi_{ij}=\frac1n\sum_{k=1}^n
 \frac{|I_i\cap U_k|}{|U_k|}\frac{|J_j\cap V_k|}{|V_k|}.
\]
Every historical pair has mass $1/n$. Ties occupy their whole rank intervals. Write $Z=E+D$. The masses $\widehat\pi_{ij}$ define the historical-rank directional probabilities on the sums $z_{ij}=e_i+d_j$. Independence uses $p_iq_j$ on the same support. The issued historical-rank interval uses the 0.05 and 0.95 generalized-inverse quantiles of its tied-rank sum distribution, shifted by $m_s$. The dependence-robust prediction interval sorts the distinct $z_{ij}$ and takes the largest lower cut $z_l$ with $\sup_{\pi\in\Pi(p,q)}P_\pi(Z<z_l)\le 0.05$ and the smallest upper cut $z_u$ with $\inf_{\pi\in\Pi(p,q)}P_\pi(Z\le z_u)\ge 0.95$. Its endpoints for the later release are $m_s+z_l$ and $m_s+z_u$.

Non-deceleration is $m_s+E+D\ge0$, and deceleration is $m_s+E+D<0$. Interval coverage includes both endpoints, while outside-tail probabilities use strict inequalities.

\par\medskip\noindent\textit{Calibrated width.}

For the independence midpoint $c_i$ and positive half-width $h_i$, let $r_i=|B_i-c_i|/h_i$. The nonconstant part of the development interval-score objective is
\[
 L(a)=\sum_{h_i>0}h_i\{2a+20(r_i-a)_+\},\qquad a\ge1.
\]
The leftmost feasible minimizer is the $h$-weighted 0.9 generalized inverse of the $r$ values, floored at one. These $h$ weights arise from the score, not chosen recency weights. Zero-width observations contribute constant losses. The factor is estimated from annual labels available at the development cutoff and then held fixed during evaluation. Development-period comparisons use the factor fitted over that period.

\par\medskip\noindent\textit{Mixture quantiles.}

Positive-scale Student components make the fitted mixture a strictly increasing continuous CDF with a unique interior inverse. Quantiles are evaluated at probabilities 0.05, 0.50 and 0.95, holding the fitted regression and component weights fixed.

The accompanying replication package documents the data sources and contains estimation and evaluation code.

\FloatBarrier

\par\phantomsection\label{arxiv-oa-lastpage}
\end{document}